\documentclass[smallextended,numbook,envcountsame]{svjour3} 
\smartqed  
\usepackage{graphicx}
\usepackage{bm}
\usepackage{amsmath,amssymb,amsfonts} 
\usepackage{mathrsfs,latexsym}
\usepackage{xcolor}
\usepackage{latexsym}
\newcommand{\CC}{{\mathbb C}}
\newcommand{\RR}{{\mathbb R}}

\newcommand{\TT}{{\mathbb T}}

\newcommand{\Cc}{{\mathcal{C}}}
\newcommand{\Dc}{{\mathcal{D}}}
\newcommand{\Ec}{{\mathcal{E}}}
\newcommand{\Fc}{{\mathcal{F}}}

\newcommand{\Hc}{{\mathcal{H}}}

\newcommand{\Kc}{{\mathcal{K}}}
\newcommand{\Lc}{{\mathcal{L}}}
\newcommand{\Mc}{{\mathcal{M}}}
\newcommand{\Nc}{{\mathcal{N}}}
\newcommand{\Oc}{{\mathcal{O}}}

\newcommand{\Pc}{{\mathcal{P}}}
\newcommand{\Qc}{{\mathcal{Q}}}

\newcommand{\Vc}{{\mathcal{V}}}

\newcommand{\ualpha}{\underline{\alpha}^c}
\newcommand{\ualphao}{\underline{\alpha}}
\newcommand{\oalpha}{\overline{\alpha}^c}
\newcommand{\oalphao}{\overline{\alpha}}

\newcommand{\bg}{{\bm g}}  
\newcommand{\biG}{{\bm G}}
\newcommand{\bih}{{\bm h}}
\newcommand{\biH}{{\bm H}}
\newcommand{\bik}{{\bm k}}
\newcommand{\biv}{{\bm v}}
\newcommand{\bix}{{\bm x}}

\newcommand{\bip}{{\bm p}}
\newcommand{\biN}{{\bm N}}

\newcommand{\biP}{{\bm P}}
\newcommand{\biQ}{{\bm Q}}

\newcommand{\fA}{{\mathfrak A}}
\newcommand{\fB}{{\mathfrak B}}

\newcommand{\fG}{{\mathfrak G}}
\newcommand{\fP}{{\mathfrak P}}
\newcommand{\fQ}{{\mathfrak Q}}

\newcommand{\Fun}{\Fc}  
\newcommand{\Pzw}{\Pc}  
\newcommand{\Qzw}{\Qc}  

\newcommand{\Alg}{\fA}                   

\newcommand{\Group}{\fG}

\newcommand{\Test}{\Dc}                  
\newcommand{\Lag}{\Lc}

\newcommand{\Hil}{\Hc}                   

\newcommand{\Reg}{\Oc}                   

\newcommand{\Jp}[1]{J_+^{\, c} #1}
\newcommand{\Jm}[1]{J_-^{\, c} #1}

\newcommand{\Jmp}[1]{J_\mp^{\, c} #1}
\newcommand{\Jcap}[1]{J_{\, \cap}^{\, c} #1}
\newcommand{\Jcup}[1]{J_{\, \cup}^{\, c} #1}
\newcommand{\perpc}{\overset{c}{\perp}}

\newcommand{\tp}{{\scriptscriptstyle +}}
\newcommand{\tm}{{\scriptscriptstyle -}}

\newcommand{\Kcc}{\Kc^{\, c}}

\newcommand{\csucc}{\overset{c}{\succ}}

\newcommand{\supp}{{\mbox{supp} \, }}

\newcommand{\be}{\begin{equation}}
\newcommand{\ee}{\end{equation}}

\def\ie{{\it i.e.\ }}
\def\viz{{\it viz.\ }}

\newcommand{\const}{\mathrm{const}}

\journalname{Annales Henri Poincar\'e}
\begin{document}

\title{Dynamical C*-algebras and kinetic perturbations}


\author{Detlev Buchholz        \and
        Klaus Fredenhagen 
}


\institute{Detlev Buchholz \at
              Mathematisches Institut, Universit\"at G\"ottingen \\
              Bunsenstr.\ 3-5, 37073 G\"ottingen, Germany \\
              \email{detlev.buchholz@mathematik.uni-goettingen.de} \\
           \and
           Klaus Fredenhagen \at
           II. Institut f\"ur Theoretische Physik, Univerist\"at Hamburg \\
           Luruper Chaussee 149, 22761 Hamburg, Germany \\
           \email{klaus.fredenhagen@desy.de}
}

\date{Received: date / Accepted: date}

\maketitle

\begin{abstract}
The framework of dynamical C*-algebras for 
  scalar fields in Mink\-ow\-ski space, based on
  local scattering operators, is extended
  to theories with locally perturbed kinetic terms. 
  These terms encode information about the underlying spacetime
  metric, so the causality relations between the 
  scattering operators have to be adjusted accordingly. It is
  shown that the extended algebra describes scalar quantum fields,  
  propagating in locally deformed Minkowski spaces.
  Concrete representations of the abstract scattering operators, inducing
  this motion, are known to exist on Fock space. The proof that these
  representers also satisfy the generalized causality relations requires,
  however, novel arguments of a cohomological nature. 
  They imply that Fock space representations of the extended dynamical
  C*-algebra exist, involving linear as well as kinetic and
  pointlike quadratic perturbations of the field.
  
  \keywords{dynamical C*-algebras \and kinetic perturbations \and
  causal phases}
\end{abstract}

\section{Introduction}

We continue here our construction of
dynamical C*-algebras for scalar quantum fields in Minkowski
space \cite{BF19}. These algebras are generated by unitary operators
$S(F)$, where $F$ denotes some real functional acting on the underlying
classical field. The classical field is described by
real, smooth functions
$x \mapsto \phi(x)$ on $d$-dimensional Minkowski space
$\Mc \simeq \RR^d$, 
and the functionals considered in \cite{BF19} were of the specific form
\be \label{e.1.1}
F[\phi] \doteq -  \sum_{j = 0}^k (1/j!) \! \int \! dx \, g_j(x) \, \phi(x)^j \, .
\ee
Here $g_j \in \Test(\Mc)$ are real test functions on $\Mc$ with compact
supports. The term for $j=0$ denotes the constant functional.
These functionals are interpreted as
perturbations of the underlying
Lagrangian by point like self interactions of the field.
Their support (in the sense of functionals) is defined as union of
the supports of the underlying test functions~$g_j\ $ for $j>0$;
the constant functional (corresponding to $j=0$) has empty support
and hence can be placed everywhere. The unitaries $S(F)$ 
are the scattering operators corresponding
to the perturbations $F$. As was shown in \cite{BF19}, 
they satisfy for a given Lagrangian a dynamical
relation, based on the Schwinger-Dyson equation, as well as the
causal factorization rule
\be \label{e.1.2}
S(F + G) S(G)^{-1} S(G+H) = S(F + G + H) \, .
\ee
This relation holds whenever the spacetime support of $F$ 
succeeds the support of $H$ with regard to the
Minkowski metric. The support of the functional $G$,
having the preceding special form, is completely arbitrary.

\medskip
In the present article we consider also
localized perturbations of the kinetic
part of the underlying Lagrangians. This is of interest if one
thinks of perturbations of the theory by gravitational forces. But
it also provides a basis for the discussion of symmetry properties
of the theory, related to Noether's theorem. The
corresponding functionals $P$ are quadratic in the 
partial derivatives of the underlying field, 
\be \label{e.1.3} 
P[\phi] 
\doteq (1/2) \int \! dx \,
\partial_\mu \phi(x) \, p^{\mu \nu}(x) \  \partial_\nu \phi(x) \, .
\ee 
Here $x \mapsto p^{\cdot \, \cdot}(x) $ are smooth
functions with compact support, regarded 
as the support of $P$, which have values in the space of real,
symmetric $d \times d$ matrices. As we shall see, 
these functions have to comply with further constraints
in order to admit a meaningful interpretation as kinetic
perturbations. To avoid the discussion of the special situation in
two dimensions, we assume $d>2$.

\medskip 
Given admissible functionals  
$P$ of this kind, we consider the corresponding
scattering operators $S(P)$. Whereas the respective dynamical relations
remain unaffected, the causal factorization rule needs
to be adapted to the particular choice of~$P$. This
can be understood if one takes into account that the 
unitary operators $F \mapsto S(P)^{-1} S(F + P)$
describe scattering processes, induced by 
functionals $F$ of the preceding types,
which evolve under the perturbed dynamics  
with perturbation given by $P$. Thus if
the functional $P$ is of kinetic type, this
scattering process effectively takes place in a 
locally distorted Minkowski space whose causal
structure, fixed by $P$, enters in
the factorization rules. Yet 
operators $S(P), S(Q)$, assigned to functionals
having their supports in spacelike separated regions of Minkowski
space, still commute. By arguments given
in \cite{BF19}, this extended family of operators therefore 
generates local nets of C*-algebras in
Minkowski space, complying with all Haag-Kastler
axioms \cite{HaKa}.

\medskip 
We will study in more detail the subalgebra of
the dynamical C*-algebra,
which is generated by scattering operators assigned to functionals
of the classical field as well as its 
kinetic and quadratic point like perturbations.
This algebra describes quantum fields in locally distorted
Minkowski spaces, which satisfy corresponding field
equations and commutation relations. We will also 
exhibit some algebraic relations between
the field and the underlying scattering operators.

\medskip 
These results enter in our construction of representations of 
this algebra on Fock space. 
In this construction we make use of the known fact that the unitary 
scattering operators associated with kinetic perturbations can be
represented on Fock space~\cite{Wald}. Yet the phase factors
of these operators remained ambiguous in that analysis.
They matter, however, for the proof that there is a 
choice such that the resulting operators satisfy the  
causal factorization relations. In order to
establish this fact, we develop arguments akin to cohomology theory. 
The existence of Fock representations of
the dynamical C*-algebra, generated by the field and its quadratic 
perturbations, is thereby established. 

\medskip 
Our article is organized as follows. In the subsequent section
we introduce notions from classical field theory  
and discuss the form of admissible kinetic perturbations. Section~3
contains the definition of the extended dynamical C*-algebra and
remarks on some
of its general properties. In Sec.~4 we study the subalgebra generated
by the field and its quadratic 
kinetic as well as point like perturbations and determine its
algebraic structure. These results are used in Sec.~5 in an   
analysis of representations of the scattering operators
and of their products on Fock space.
The ambiguities left open in the phase factors are discussed
in Sec.~6; there it is shown that, for some coherent choice of these
factors, the scattering operators satisfy the causal
factorization rules and thus define a representation of the
C*-algebra on Fock space. The article concludes with
a brief outlook and a technical appendix. 

\section{Classical field theory}

We adopt the notation used in \cite{BF19} and adjust it to the 
more general setting, considered here. As already
mentioned, we proceed from a classical scalar field on $d$-dimensional 
Minkowski space $\Mc \simeq \RR^d$ with its standard metric
$\eta(x,x) = x_0^2 - \bix^2$, where $x_0, \bix$ denote the time and
space components of $x \in \RR^d$. 
The field is described by real, smooth functions
$x \mapsto \phi(x)$, which constitute its  
configuration space $\Ec$. 
The Lagrangian density of a non-interacting field
with mass $m \geq 0$ is given by
\be \label{e.2.1}
x \mapsto \Lag_0(x)[\phi] =  1/2 \,
(\partial_\mu \phi(x) \, \eta^{\mu \nu} \, \partial_\nu \phi(x) 
- m^2 \, \phi(x)^2 ) \, .
\ee
Its spacetime integral (if defined) is the corresponding Lagrangian action.
The passage to fields which are subject to interaction,    
as given in \eqref{e.1.1} or \eqref{e.1.3}, is accomplished by
adding to this Lagrangian the respective densities.

\medskip 
On the configuration space $\Ec$ of the field acts the 
additive group $\Ec_0$ of deformations, described by 
test functions $\phi_0 \in \Test(\Mc)$.
Their action on the affine space $\Ec$ is given by local shifts
of the field, 
$\phi \mapsto \phi + \phi_0$. With their help one 
defines variations of the Lagrangian action functionals, given by  
\be \label{e.2.2}
\delta \Lag(\phi_0)[\phi]
\doteq \int \! dx \, \big(\Lag(x)[\phi + \phi_0] - \Lag(x)[\phi]\big) \, .
\ee
These variations are well defined for local Lagrangians
and arbitrary fields $\phi$ in view of the
compact support of $\phi_0$. Their stationary points 
define the solutions of the classical field equation for
the given Lagrangian (``on shell fields''). 

\medskip
In case of the non-interacting Lagrangian \eqref{e.2.1}, 
the corresponding on shell field satisfies the Klein-Gordon equation.
If one adds to this Lagrangian the
densities of a kinetic perturbation $P$ as in
\eqref{e.1.3} and of a quadratic perturbation  $F_2$ 
with potential $g_2 = q$ as in \eqref{e.1.1},
the resulting field equation reads 
\be \label{e.2.3} 
\partial_\mu \, (\eta^{\mu \nu} + p^{\mu \nu}(x)) \, \partial_\nu \, \phi(x)
+ (m^2 + q(x)) \phi(x) = 0 \, .
\ee
We restrict our attention here to perturbations $P$ for which 
this equation describes the propagation of the field $\phi$ on a globally
hyperbolic spacetime with metric $g_P$. This metric is, up to a factor, 
the inverse of the principal symbol \cite{Hoer} of the
underlying differential operator, $x \in \Mc$,
\be \label{e.2.4}
|\mathrm{det}g_P(x)|^{-1/2}g_P(x) \doteq  (\eta + p(x))^{-1} \, .
\ee
In order to simplify the discussion of the causal
factorization relations, we restrict our attention to
metrics $g_P$ for which (a) the constant time planes of
Minkowski space for some fixed time coordinate are
Cauchy surfaces and (b) the time coordinate is  
positive timelike with regard to all of these metrics.
As is shown in the appendix, it amounts to the following condition.

\medskip \noindent
\textbf{Standing assumption:} The coefficients
$p^{\mu \nu}(x)$, $\mu, \nu = 0, \dots , d-1$,
\mbox{$x \in \Mc$}, of the kinetic
perturbations $P$ are smooth functions with compact support
which satisfy
\begin{itemize}
\item[(i)] $1 + p^{00}(x) > 0$,

\vspace*{1mm}
\item[(ii)]
  the matrix 
  $\, \delta^{ij} + p^{ij}(x) \,$
is positive definite, $\, i,j = 1, \dots , d-1$. 
\end{itemize}
The family of kinetic perturbations satisfying this condition
is, for each \mbox{$x \in \Mc$}, convex and stable under scalings by 
positive numbers which are bound\-ed by~$1$, cf.\ the appendix.
It is also invariant under spacetime translations. 
In view of the choice of a distinguished time coordinate
underlying its definition, it is, however,
not Lorentz invariant. We will return to this point in
subsequent sections. 

\section{The extended dynamical algebra}

The functionals $F: \Ec \rightarrow \RR$ considered
in this section contain, in addition to point like interactions as 
in equation \eqref{e.1.1}, kinetic perturbations~\eqref{e.1.3} 
with properties specified in the standing assumption.
The family of these functionals is denoted by $\Fun$. 
Whereas $\Fun$ is, in general, not stable under addition, we will
deal with special pairs and triples of functionals in $\Fun$ for which
all (partial) sums satisfy the standing assumption. Such
tuples will be termed \textit{admissible}. 

\medskip 
Apart from the spacetime localization 
of the functionals, fixed by the supports of the underlying
test functions, we must also take into account their impact 
on the causal structure of spacetime.
For $P \in \Fun$, this structure is determined by the metric $g_P$,
which is fixed by the kinetic part of~$P$ according to
equation \eqref{e.2.4}. Given any region $\Oc \subset \RR^d$,
we denote by $J^P_\pm(\Oc)$ the causal future, respectively past,
of $\Oc$ with regard to $g_P$.
In case of the Minkowski metric, $P = 0$, we 
write $J_\pm^{\, 0}(\Oc)$.

\medskip 
Given an admissible triple $P, Q, N \in \Fun$, we say that 
$P$ succeeds $Q$ with regard to (the metric induced by)~$N$ if 
$\mbox{supp} \, P$ does not intersect the past cone of  
$\mbox{supp} \, Q$, determined by $g_N$, \ie 
$\mbox{supp} \, P \cap J_-^N(\mbox{supp} \, Q) = \emptyset$. 
In this case we write $P \underset{N}{\succ} Q$.  
In particular, $P \underset{0}{\succ} Q$
means that $P$ succeeds $Q$ in Minkowski space.
Note that $\succ$ is \textit{not} an ordering relation,
in particular it is not transitive. 
Based on these notions, we can proceed now to 
an extension of the dynamical algebras, introduced in \cite{BF19},
by adding  to them the kinetic perturbations. 
As in \cite{BF19}, we
begin by defining a dynamical group, generated
by symbols $S(P)$, $P \in \Fun$, which are subject to two
relations. These relations involve a given Lagrangian $\Lag$,
the corresponding relative action~\eqref{e.2.2}, and
shifts of the functionals $P$ by elements $\phi_0 \in \Ec_0$,
denoted by $P^{\phi_0}[\phi] \doteq P[\phi + \phi_0]$, $\phi \in \Ec$. 
Compared to \cite{BF19}, we employ here a somewhat
simplified ``on shell'' version of this group.

\medskip 
\noindent \textbf{Definition:} \ Given a local Lagrangian $\Lag$
on Minkowski space $\Mc$, the corresponding
dynamical group $\Group_\Lag$ is the free group generated by 
elements $S(P)$, $P \in \Fun$, with  $S(0) = 1$, modulo the relations
\begin{itemize}
\item[(i)] $S(P) = S(P^{\phi_0} + \delta \Lag(\phi_0))$ \ for \
  $P \in \Fun$, $\phi_0 \in \Ec_0 \,$,
\item[(ii)] $S(P + N) S(N)^{-1} S(Q + N) = S(P + Q + N)$
\ for any admissible triple $P,Q,N \in \Fun$ such that $P$ succeeds 
$Q$ with regard to $N$, \ $P \underset{N}{\succ} Q$. 
\end{itemize} 

\noindent \textbf{Remark:}  If one puts $N=0$ in the second
condition, one obtains the causality relation
$S(P) S(Q) = S(P +Q)$ if $P$ succeeds $Q$ with regard
to the Minkowski metric. Thus if 
$P$, $Q$ have spacelike separated supports in Minkowski
space, then also $S(Q) S(P) = S(Q + P)$ and the
operators commute. 

\medskip 
A thorough discussion of the origin and interpretation of these
relations is given in \cite{BF19,BF20}. 
The only difference with regard to the present framework
appears in relation (ii), where the
impact of the kinetic functionals on the causal structure
of spacetime is taken into account. 

\medskip
The passage from the dynamical group $\Group_\Lag$
to a corresponding C*-algebra  
is accomplished by standard arguments, cf.~\cite{BF19}.
One regards 
the elements of $\Group_\Lag$ as basis of some complex vector
space $\Alg_\Lag$; the product in $\Alg_\Lag$ is inherited from
$\Group_\Lag$ by the distributive law, and the
*-operation can be defined such that the generating
elements $S(P)$ become unitary operators.
The resulting *-algebra has faithful states
and thus can be equipped with a (maximal) C*-norm.
Its completion defines the dynamical C*-algebra $\Alg_\Lag$ 
for given Lagrangian $\Lag$ and generating operators  
$S(P)$, $P \in \Fun$, describing local operations on the
underlying system. 

\medskip
A distinguished role is played by the constant functionals
which, for $c \in \RR$, are given by $c[\phi] \doteq c$, $\phi \in \Ec$.
Their support is empty, hence
\be
S(c) S(P) = S(c + P) = S(P) S(c)
\ee
by the causality condition (ii). So $c \mapsto S(c)$ 
defines a unitary group in the center of $\Alg_\Lag$. As in
\cite{BF19}, we fix its scale and put $S(c) = e^{ic} \, 1$, $c \in \RR$. 

\medskip
In a similar manner, one can define extended dynamical
algebras for theories on arbitrary globally hyperbolic
spacetimes. There the admissible kinetic perturbations
need to be adjusted to the underlying metric.
We restrict our attention here to Minkowski space and its
local deformations, inherited from functionals in $\Fun$.
For given $M \in \Fun$, these perturbations can still
be described by unitary operators in the algebra~$\Alg_\Lag$.
As brought to light by Bogoliubov \cite{BP57,BS59}, they  
are defined by  
\be \label{e.3.1} 
S_M(P) \doteq S(M)^{-1} S(M + P) \, , \quad P \in \Fun \, .
\ee
One easily verifies that these operators also satisfy the
two defining relations of some dynamical algebra. 
In the first relation, 
the Lagrangian $\Lag$ is to be replaced by $\Lag_M$, \ie 
the Lagrangian obtained from $\Lag$ by adding to it
the density inherent in $M$. The factorization equation in
the second relation is satisfied for
admissible quadruples $P,Q,N,M \in \Fun$,
provided $P$ succeeds $Q$ with regard to
$(M + N)$, \ie $P \underset{(M + N)}{\succ} Q$. 

\section{Quadratic perturbations}

We take from now on as dynamical input 
the algebra $\Alg$ for the Lagrangian~$\Lag_0$, cf.~\eqref{e.2.1}, 
omitting in the following the subscript $\Lag_0$. In fact, we 
are primarily interested in its 
subalgebra $\Alg_2 \subset \Alg$, which is
generated by unitaries $S(P)$ with functionals
$P \in \Pzw$, where $\Pc \subset \Fun$ denotes the family
of functionals which are at most quadratic in the 
underlying field and satisfy our
standing assumption; its subset of genuine quadratic functionals is 
denoted by $\Qzw$. 
As we shall see, the algebra $\Alg_2$ 
comprises non-interacting quantum fields, propagating in
locally deformed Minkowski spaces. 

\medskip
We adopt the notation used in \cite{BF19}. Thus
$K \doteq -(\partial_\mu \, \eta^{\mu \nu} \, \partial_\nu + m^2)$
is the negative Klein-Gordon operator,
$\Delta_R$ and $\Delta_A$ are the corresponding
retarded and advanced propagators, their difference
$\Delta = (\Delta_R - \Delta_A)$ is the commutator
function, and $\Delta_D = (1/2) (\Delta_R + \Delta_A)$
is the Dirac propagator. Further below, we will
also introduce perturbed versions of these entities. 

\medskip
As in \cite{BF19}, we consider perturbations involving linear functionals of 
the fields $\phi \in \Ec$, given by 
\be \label{e.4.1} 
F_f[\phi] = L_f[\phi] + (1/2) \, \langle f, \Delta_D f \rangle \, ,
\quad f \in \Test(\Mc) \, .
\ee
Here $L_f[\phi] \doteq \int \! dx \, f(x) \, \phi(x)$ and
$\langle f, g \rangle \doteq \int \! dx \, f(x) \, g(x)$
are constant functionals, where $f,g$ are smooth
functions whose pointwise product  $fg$ is compactly supported. 
It was shown in \cite{BF19} that the unitary operators
\be \label{e.4.2} 
W(f) \doteq S(F_f)
= S(L_f)    \, e^{(i/2) \langle f, \Delta_D f \rangle}
\in \Alg_2 \, , \quad f \in \Test(\Mc) \, ,
\ee
have the algebraic properties of Weyl operators on Minkowski space.
In particular, 
\be
W(Kf) = 1 \, , \ \ 
W(f) W(g) = e^{-(i/2) \langle f, \Delta g \rangle} W(f + g) \, ,  \ \ 
f,g \in \Test(\Mc) \, .
\ee
So these operators can be interpreted as
exponential functions of a quantum field,
which satisfies the
Klein-Gordon equation and has c-number commutation relations given
by the commutator function $\Delta$. 

\medskip
Next, we compute the product of Weyl operators
with arbitrary elements of the
full algebra $\Alg$. The result is stated in the
following lemma. There we make use again of the shift of functionals
by elements of~$\Ec_0$. As a matter of fact, taking
advantage of the support properties of the functionals, these shifts
are canonically extended in the lemma to a larger family
of smooth functions. 

\begin{lemma} \label{l.4.1}
Let $P \in \Fun$ and let $f \in \Test(\Mc)$. Then
\begin{enumerate}
\item[(i)]  $W(f) S(P) = S(F_f + P^{\, \Delta_R f})$, \quad
  $S(P) W(f) = S(F_f + P^{\, \Delta_A f})$
\item[(ii)] $W(f) S(P) W(f)^{-1} = S(P^{\, \Delta f})$. 
\end{enumerate}  
The condition of associativity does not  
entail further relations for multiple products of Weyl
operators with operators $S(P)$. 
\end{lemma}
\begin{proof}
  To compute $W(f) S(P)$, we decompose 
  $f$ into $f = f_P + K g_P$, where $f_P,g_P$ are test functions 
  and the support of $f_P$ succeeds that of $P$ with regard to
  the Minkowski metric, cf.\ \cite[Sec.\ 4]{BF19}.   Thus
  \mbox{$W(f) = W(f_P)$}, hence, making use of the
  causal factorization condition as well as the dynamical
  relation underlying $\Alg$, we obtain 
  \begin{align}
    W(f) S(P) & = W(f_P) S(P) = S(F_{f_P} + P) \nonumber \\
    & = S\big(F_{f_P}^{g_P} + P^{g_P} + \delta \Lag_0(g_P)\big) \, . 
  \end{align}
  By an elementary computation one finds that
  $F_{f_P}^{\, g_P} +  \delta \Lag_0(g_P) = F_f$. Since the support of $f_P$,
  whence that of $\Delta_R f_P$, succeeds that of 
  $P$ and
  \be
  g_P = \Delta_R K g_P = \Delta_R \, (f - f_P) \, ,
  \ee
  one has $P^{\, g_P} = P^{\, \Delta_R f}$. Thus we arrive at
  $W(f) S(P) = S(F_f + P^{\, \Delta_R f})$. In an analogous manner
  one obtains the second equality in the first part of the
  statement.

\medskip 
  As to the second part, we make use of $W(f)^{-1} = W(-f)$, giving 
    \begin{align}
      \big(W(f) S(P) \big)
      W(-f) & = S(F_f + P^{\Delta_R f}) \,  W(-f) \nonumber \\
                  & = S(F_{-f} + F_f^{-\Delta_A f} + P^{\, (\Delta_R - \Delta_A) f}) \, .
    \end{align}
    Since the commutator function $\Delta = \Delta_R - \Delta_A$
    is antisymmetric, the first two functionals in the latter
    operator compensate each other, \viz
    \begin{align}
    F_{-f} + F_f^{- \Delta_A f} = \langle f, \Delta_D f \rangle
    - \langle f, \Delta_A f \rangle = (1/2) \langle f, \Delta f \rangle = 0
    \, ,
  \end{align}    
    proving statement (ii). 
    
    \medskip
    It remains to establish the assertion about multiple products.
    Picking any $f,g \in \Test(\Mc)$, it follows from the Weyl relations
    and the preceding step that
    \begin{align}
    \big(W(f) W(g)\big) \, S(P) & = e^{-(i/2) \langle f, \Delta g \rangle}
      S(F_{f + g} + P^{\Delta_R (f + g)}) \nonumber \\
      & = S(F_{f + g} -
      (1/2) \langle f, \Delta g \rangle + P^{\Delta_R (f + g)})) \, .
    \end{align}
    On the other hand, interchanging brackets,  one obtains 
    \begin{align}
      W(f) \, \big( W(g) S(P) \big) & =
      W(f) \, S(F_g + P^{\Delta_R g}) \nonumber \\
      & = S(F_f + F_g^{\Delta_R f} + P^{\Delta_R(g + f)}) \, .
    \end{align}
    By another elementary computation, one verifies that 
    \be F_{f + g} - 
      (1/2) \langle f, \Delta g \rangle =
    F_f + F_g^{\Delta_R f} \, ,
    \ee
    hence the operators in the preceding two relations coincide.
    In a similar manner one sees that also all other products do not
    produce any new relations. \qed \end{proof}

We turn now to the analysis of the subalgebra $\Alg_2 \subset \Alg$.
Its generating elements $S(\biP)$ are given by functionals
of the form 
\be 
\biP \doteq (P_0 + P_1 + P_2) \in \Pzw \, ,
\ee
where $P_0$ is constant,
$P_1$ is linear, and $P_2$ is quadratic in the underlying field.

\medskip 
Given a functional $P_2 \in \Qzw$,
we consider perturbations of the Lagrangian $\Lag_0$
by adding to it the density $P$ of 
$P_2[\phi] = (1/2) \langle \phi, P \phi \rangle$, $\phi \in \Ec$.
The perturbed Lagrangian is denoted by $\Lag_P$ and the 
resulting classical field equation~\eqref{e.2.3} involves 
the differential operator $-(K + P)$. 
As is well known, cf.\ for example~\cite{BGP}, there exist
corresponding retarded and
advanced propagators $\Delta_R^P$ and $\Delta_A^P$, fixing 
the commutator function $\Delta^P \doteq (\Delta_R^P - \Delta_A^P)$,
and the Dirac propagator
\mbox{$\Delta_D^P \doteq (1/2) (\Delta_R^P + \Delta_A^P)$}.
In view of the regularity properties of $P$, these
distributions map test functions into smooth functions. 
We will frequently make use of the basic relation
on $\Test(\Mc)$ 
\be
(K + P) \, \Delta_{A, R}^P = \Delta_{A,R}^P \, (K+P) = 1
\ee
and the resolvent equation 
\be \label{e.4.12}                        
\Delta_{A, R}^P - \Delta_{A, R}
= -  \Delta_{A, R}^P \, (P \Delta_{A, R}) =
- \Delta_{A, R} \, (P \Delta_{A, R}^P) \, .
\ee 
These relations hold on the test functions $\Test(\Mc)$. 
Note  that $(P \Delta_{A, R})$ and
$(P \Delta^P_{A, R}) = (1 - K \Delta^P_{A, R})$  
map test functions into test functions. 

\medskip 
The analysis of the properties of the operators
$S(\biP)$, $\biP \in \Pzw$, simplifies by 
making use of the fact that the contributions coming from
the constant and linear functionals
$P_0$ and~$P_1$ can be factored out from $S(\biP)$. For
constant functionals, this was already shown in the
preceding section. For the linear functionals, introduced
above, this is
a consequence of the preceding lemma. Namely, making use 
of the quadratic dependence of~$P_2$ on the field, one obtains   
\be \label{e.4.10}
F_f + P_2^{\, \Delta_A f} =
L_{(K + P) \Delta_A f}
+ (1/2) \langle \Delta_A f, (K + P) \Delta_A f \rangle + P_2  \, .
\ee
Thus, by the preceding lemma and the definition of Weyl operators, 
\be \label{e.4.11}
S(P_2) S(L_f) \, e^{(i/2) \langle f, \Delta_D f \rangle} = 
S(L_{(K+P) \Delta_A f} + P_2) \,
e^{(i/2) \langle \Delta_A f, (K + P) \Delta_A f \rangle } \, .
\ee
Noticing that the inverse of $(K + P) \Delta_A$ is given by
$K \Delta_A^{P}$, one sees  
that the linear functionals can be extracted 
from the operators $S(\biP)$, as well. We may therefore restrict 
our attention in the following to quadratic
perturbations $P_2 \in \Qzw$ and omit the index $2$. Without danger
of confusion, we will also
equate these perturbations with their respective densities. 

\medskip 
Given a perturbation $P \in \Qzw$, the perturbed algebra
$\Alg_{\Lag_P} \subset \Alg$ for the Lagrangian $\Lag_P$ is generated by the
unitary operators, cf. equation~\eqref{e.3.1}, 
\be
S_P(Q) \doteq  S(P)^{-1} S(P + Q) \, , \quad Q \in \Qzw \, . 
\ee
Defining, in analogy to \eqref{e.4.1}, functionals 
$
F^P_f[\phi] \doteq L_f[\phi] + (1/2) \langle f, \Delta^P_D f \rangle
$
on $\Ec$, it turns out that the corresponding perturbed operators
\be  \label{e.4.8} 
W_P(f) \doteq S_P(F^P_f) \, , \quad f \in \Test(\Mc) \, ,
\ee 
coincide with the Weyl operators for perturbed test
functions. In fact, according to relation \eqref{e.4.10}
we have $F_f + P^{\Delta_A f} = F^P_{(K + P) \Delta_A f} + P$.
Hence, making use of the lemma and the fact that
$\big((K + P) \Delta_A\big)^{-1} = K \Delta_A^P$, we arrive at
\be \label{e.4.14} 
W_P(f) = W(K \Delta_A^P f) \, ,  \quad f \in \Test(\Mc) \, .
\ee 
The perturbed operators $W_P(f)$, $f \in \Test(\Mc)$,
describe the exponential function
of a quantum field which satisfies a linear field equation
with regard to $K+P$. This follows from 
\be
W_P((K+P)f) = W(K \Delta_A^P (K+P) f)
= W(K f) = 1 \, .
\ee
Moreover, they satisfy the Weyl
relations with respect to the commutator function $\Delta^P$ 
fixed by $(K+P)$. In order to verify this we need to
compute the symplectic form 
$\langle (K \Delta_A^P f), \Delta \, (K \Delta_A^P g) \rangle$
for $f, g \in \Test(\Mc)$. Bearing in mind the
properties of propagators, mentioned above, we have  
\begin{align}
& \langle \Delta_A  (1-P \Delta^P_A) \, f, \, (1-P\Delta_A^P) \, g \rangle
  = \langle \Delta_A^P \, f, \, g \rangle -
  \langle  \Delta_A^P \, f \, , P \Delta_A^P \, g \rangle \, ,
  \nonumber \\[1mm] 
&  \langle (1-P \Delta^P_A) \, f, \Delta_A (1-P \Delta^P_A) \, g \rangle 
  = \langle f, \Delta_A^P \, g \rangle - \langle P  \Delta_A^P \, f,
\Delta_A^P \, g \rangle \, .
\end{align}
Since $P$ is compactly supported, it acts as a symmetric operator on
smooth functions, so the last terms in the preceding two 
equalities coincide. We therefore obtain 
\begin{align}
  & \langle (K \Delta_A^P f), \Delta \, (K \Delta_A^P g) \rangle
  = \langle(1-P\Delta^P_A) f, \, \Delta \, (1-P\Delta_A^P) \,
  g \rangle  \nonumber \\[1mm]
  & = \langle \Delta_A (1-P\Delta^P_A) \, f, \, (1-P\Delta_A^P) \, g\rangle
  - \langle (1-P\Delta^P_A) \, f, \, \Delta_A (1-P\Delta_A^P) \,
  g \rangle \nonumber \\[1mm]
  & = \langle \Delta_A^P \, f, \, g \rangle -
  \langle f, \, \Delta_A^P \, g \rangle
  = \langle f, \, \Delta^P \, g \rangle \, .
\end{align}
Thus we arrive at the Weyl relations for the perturbed operators,
\be
W_P(f) \, W_P(g) = e^{-(i/2) \langle f, \Delta^P g \rangle} \, W_P(f+g) \, ,  \quad 
f,g \in \Test(\Mc) \, .
\ee

It follows from this equality  that the commutation 
relations of the operators in $\Alg_{\Lag_P}$, $P \in \Qzw$,  depend 
on the causal structure induced by the principal symbol of $(K+P)$.
On the other hand, the  perturbative computation    
of its generating elements $S_P(Q)$, $Q \in \Pzw$, 
inherits the causal structure of Minkowski space
\mbox{\cite{Du,EpGl,Re}}. Thus  
the perturbative expansion of these operators
will in general not converge.

\section{Construction of Fock representations}

Whereas for Weyl operators the existence of Fock representations is
a well known fact, the question of whether these representations can
be extended to the full dynamical algebras involving arbitrary
local interactions is an open problem. As
a matter of fact, this question may be regarded as the remaining fundamental
problem of constructive quantum field theory \cite{BF19}.
We therefore restrict our attention here to the algebra~$\Alg_2$, involving 
perturbations of the non-interacting Lagrangian which are
at most quadratic in the underlying field. Even there, the question 
of whether this algebra is represented on Fock space
has remained open to date, to the best of our knowledge.

\medskip
In order to discuss this problem, we adopt the following strategy:
proceeding from a representation of the Weyl algebra on Fock space,
we make use of the fact that the quadratic perturbations induce
automorphisms of this algebra. It then follows from a 
result by Wald \cite{Wald} that
these automorphisms can be unitarily  
implemented on Fock space. In the present section we complement
this result by the observation that the automorphisms satisfy
an automorphic version of the causal factorization condition.
Since the Weyl algebra is irreducibly represented on Fock space,
this implies that the implementing unitary operators satisfy
the factorization condition, up to phase factors. 
In the subsequent section we will then show that the
phase of the unitary operators can be adjusted such that they fully
comply with causal factorization. 

\medskip
The computation of the adjoint action of quadratic 
perturbations $S(P)$   
on the Weyl operators, $P \in \Qzw$, 
is accomplished with the help of Lemma~\ref{l.4.1}. It yields,
$f \in \Test(\Mc)$, 
\begin{align}
& S(P)^{-1} W(f) S(P)  = S(P)^{-1} S(F_f + P^{\Delta_R f}) \nonumber \\[1mm] 
 & = S(P)^{-1} S(F^P_{(K+P) \Delta_R f} + P) = W_P((K+P) \Delta_R f) \, .
\end{align}
In the second equality, we made use of
equation \eqref{e.4.10}, where   
$\Delta_A f$ has been replaced by $\Delta_R f$ and, in the last
equality, we employed definition \eqref{e.4.8}
of the perturbed Weyl operators. According to relation \eqref{e.4.14}, the
latter operator coincides with $W((K \Delta_A^P) ((K+P) \Delta_R) f)$, 
where the products of
propagators and differential operators in
the brackets preserve the domain $\Test(\Mc)$. 
Noticing that $(K \Delta_A^P) ((K+P) \Delta_R)$
has an inverse given by
\be \label{e.5.2}
T_P \doteq (K \Delta_R^P) ((K+P) \Delta_A) =
(1 - P \Delta_R^P) (1 + P \Delta_A) \, ,
\ee
we arrive at 
\be \label{e.5.3}
S(P) W(f) S(P)^{-1} = W(T_P f) \, , \quad f \in \Test(\Mc) \, . 
\ee
One easily verifies that $T_P$ acts as the identity on $K \Test(\Mc)$,
hence it defines a real linear operator
on the quotient space $\Test(\Mc) / K \Test(\Mc)$. It
also follows from the preceding equality that it  
preserves the symplectic form, entering in the Weyl
relations, which is given by the
commutator function $\Delta$. So it is an invertible
symplectic transformation on the symplectic space
$\Test(\Mc) / K \Test(\Mc)$.

\medskip
This quotient space is
canonically associated with the Fock space of a particle. 
We denote by $\Hil$ the symmetric Fock space, based on the single
particle space $\Hil_1$ of a particle with mass $m \geq 0$.
The scalar product in $\Hil_1$ is fixed by
\be \label{e.5.4}
\big( f, g \big) \doteq
\int \! dp \ \theta(p_0) \delta(p^2 - m^2) \, 
\overline{\widetilde{f}(p)} \, \widetilde{g}(p) \, ,
\quad f,g \in \Test(\Mc) \, .
\ee
So the quotient $\Test(\Mc)/K \Test(\Mc)$ can be identified
with the dense subspace of $\Hil_1$, given by the
restrictions of the Fourier transforms
$\widetilde{f}$ of the test functions
to the mass shell $p^2 = m^2$, $p_0 \geq 0$.
Moreover, the imaginary part of the scalar product
in \eqref{e.5.4} coincides with the symplectic
form~$\langle f, \Delta \, g \rangle$, $f,g \in \Test(\Mc)$. 

\medskip
It follows that the operator on 
$\Test(\Mc)/K \Test(\Mc)$, fixed by $T_P$, acts as a
real linear, symplectic, and invertible operator $\overline{T}_P$
on a dense domain in the single particle space $\Hil_1$.
In fact, as was shown by Wald, 
the operator $\overline{T}_P$ is bounded \cite[Sec.\ 2.1]{Wald}.
Denoting by $\overline{T}_P^{\, \dagger}$ the adjoint of $\overline{T}_P$
with regarded to the scalar product given by the real part of
\eqref{e.5.4}, Wald also showed that 
the difference $(\overline{T}_P^{\, \dagger} \overline{T}_P - 1)$
lies in the Hilbert-Schmidt class \cite[Sec.\ 3]{Wald}.
This is a consequence of the fact that its
kernel $D_P$ can be represented as difference of
two Hadamard bi-solutions of the Klein Gordon equation,
\ie as a smooth bi-solution, 
\be
\mbox{Re} \, \big((T_P f, T_P g) - (f,g)\big)
= \iint \! dx dy \, f(x) D_P(x,y) \, g(y) \, .
\ee
Moreover, since $(T_P - 1)f$, $f \in \Test(\Mc)$, are test
functions, having their supports in the support of $P$, the kernel 
$D_P$ vanishes rapidly in spatial
directions if $m > 0$. Hence it determines
a Hilbert-Schmidt operator on~$\Hil_1$. If $m=0$, this still holds true 
in spacetime dimensions  $d \geq 4$.

\medskip
As shown by Shale \cite{Shale}, these facts imply that 
the automorphisms of the Weyl algebra, given in \eqref{e.5.3},
can be unitarily implemented on Fock space. Since the Weyl
operators act irreducibly on this space, these unitary
implementers are fixed, up to some phase factor.
The determination of these factors will occupy us in
the subsequent section. For the sake of simplicity, we
keep the notation $S(P)$ 
for the concrete Fock space representations
of the abstractly defined operators.
In the next step we show that the 
symplectic operators~$T_P$, underlying
their definition, satisfy a causal
factorization relation.

\medskip
Let $Q \in \Qzw$ and let
$g \doteq (K+Q) \Delta_A \, f$ with $f \in \Test(\Mc)$. 
Since $(K+Q)$ is a normally hyperbolic differential operator, 
there exist test functions $g_Q, h_Q$ such that 
\be \label{e.5.6} 
g = g_Q + (K+Q) \, h_Q
\ee 
and $\, \mbox{supp} \, g_Q \, \cap \, J^{\, 0}_-(\mbox{supp} \, Q) = \emptyset$. 
In fact, one can put $g_Q = (K+Q) \, \chi \, \Delta_R^Q g$,
$h_Q = (1 - \chi) \Delta_R^Q \, g$, where $\chi$
is a smooth function which vanishes in a neighborhood
of $J^{\, 0}_-(\mbox{supp} \, Q)$ and is equal to $1$ in the
complement of a slightly larger neighborhood. 
Because of the support properties of $g_Q$, 
one has
\be
(\Delta_R^Q - \Delta_R) \, g_Q = - \Delta_R^Q (Q \Delta_R) \, g_Q = 0 \, ,
\ee
hence 
\be \label{e.5.7}
T_Q f = (K \Delta_R^Q) (g_Q + (K+Q) \, h_Q)
= g_Q + K h_Q \, .
\ee
If $\mbox{supp} \, g \cap J^{\, 0}_-(\mbox{supp} \, Q) = \emptyset$,
there exists by the preceding argument a decomposition such that also   
$\mbox{supp} \, h_Q \cap J^{\, 0}_-(\mbox{supp} \, Q) = \emptyset$. 

\medskip 
Let us assume now that the pair $P,Q \in \Qzw$ is admissible
and that the support of $P$ succeeds that of  
$Q$ in Minkowski space, \ie $P \underset{0}{\succ} Q$. 
We choose an open neighborhood $\Cc$ of some Cauchy surface
in $\Mc$ which lies between $P$ and $Q$, \ie 
\be
J^{\, 0}_+(\mbox{supp} \, P) \cap \Cc = J^{\, 0}_-(\mbox{supp} \, Q) \cap \Cc 
= \emptyset \, . 
\ee
Let $f \in \Test(\Mc)$ with $\mbox{supp} \, f \subset \Cc$.
Then $\, \mbox{supp} \, T_Q f \subset J^{\, 0}_-(\Cc)$ and 
there is a decomposition \eqref{e.5.7}
such that $\, \mbox{supp} \, g_Q \subset \Cc$ and
$\, \mbox{supp} \, h_Q \cap \mbox{supp} \, P  = \emptyset$.
Thus $\, P \Delta_A \, g_Q = P \, h_Q = 0$. Since 
$\Delta_R^P g_Q$ has support in the complement of
$J_-^{\, 0}(\text{supp} \, Q)$, whence 
$\, (\Delta_R^{P + Q} - \Delta_R^P) \, g_Q
= - \Delta_R^{P + Q} (Q \Delta_R^P) \, g_Q = 0 \,$, 
it follows that 
\begin{align}
  T_P T_Q f & = (K \Delta_R^P) ((K+P) \Delta_A) \,
  (g_Q + K h_Q) \nonumber \\
  & = (K \Delta_R^P) \, (g_Q + (K+P) \, h_Q) = K \Delta_R^{P + Q} g_Q + K h_Q
  \, . 
\end{align}
According to  relation \eqref{e.5.6} 
\begin{align}
  g_Q & = g - (K+Q) h_Q  \nonumber \\ 
  & = (K+Q) (\Delta_A f - h_Q) = (K+P+Q)(\Delta_A f - h_Q) \, ,
\end{align}
so we obtain 
\be
T_P T_Q f = K \Delta_R^{P + Q}((K+P+Q)(\Delta_A f - h_Q)) + K h_Q
= T_{P + Q} \, f \, . 
\ee
Since any test function $f$ can be represented in the form
$f = f_\Cc + K  g_\Cc$ with $\mbox{supp} \, f_\Cc \subset \Cc$
and the operators $T_P, T_Q$ and $T_{P + Q}$ act on the image of
$K$ as the identity, the preceding relation holds for all
$f \in \Test(\Mc)$. Thus we have arrived at the causal
factorization relation in Minkowski space 
\be \label{e.5.12}
T_P T_Q = T_{P + Q} \, , \quad P \underset{0}{\succ} Q \, . 
\ee

We turn now to the general case. Let
$P, Q, N$ be
an admissible triple of
quadratic perturbations such that $P$ succeeds $Q$
with regard to $N$. Putting 
$T^N_P \doteq T_N^{-1} T_{P + N}$, we need to show that 
\be \label{e.5.13}
T^N_P T^N_Q = T^N_{P + Q} \quad \text{if} \quad
P \underset{N}{\succ} Q \, . 
\ee
For the metric $g_N$, fixed by $N$, there 
exists an open neighborhood $\Cc$ of some Cauchy surface
in $\Mc$ such that 
\be
J^N_+(\supp P) \cap \Cc =J^N_-(\supp Q) \cap \Cc = \emptyset \, .
\ee
Turning to the proof of the causality relation,  we proceed from 
\begin{align} \label{e.5.15}
& T_Q^N (K \Delta_A^N)  \nonumber \\
& = \underbrace{(K \Delta_A^N)(( K+N) \Delta_R)}_{T_N^{-1}} \, 
\underbrace{(K \Delta_R^{N+Q}) ((K+N+Q) \Delta_A)}_{T_{N+Q}} (K \Delta_A^N)
\end{align}
Now $\Delta_{A,R} \, (K \Delta_{A,R}^N) = \Delta_{A,R}^N $, as a 
consequence of the resolvent equation~\eqref{e.4.12}. Hence the
preceding equality simplifies to
\be
T_Q^N (K\Delta_A^N)=(K\Delta_A^N)\big((K+N)\Delta_R^{N+Q} \big)
\big( (K+N+Q)\Delta_A^N \big) \, . 
\ee
We observe that after a similarity transformation with $K\Delta_A^N$,
the operator $T_Q^N$ has the same form as $T_Q$ with the Klein Gordon
operator $K$ replaced by $(K+N)$. Thus the argument for the product rule
\eqref{e.5.13} is the same as for \eqref{e.5.12}, noticing that all
underlying propagators have support properties which are
consistent with the causal order relative to the
chosen broadened Cauchy surface $\Cc$. Multiplying
equation \eqref{e.5.13} from the left by $T_N$, we
arrive at
\be
T_{P + N} T_N^{-1} T_{Q + N} = T_{P + Q + N} \quad \text{if}
\quad P \underset{N}{\succ} Q \, .
\ee
This equality implies that the adjoint action of
$S(P + N) S(N)^{-1} S(Q + N)$ on Weyl operators coincides
with the action of $S(P + Q + N)$. So these two operators
comply with the condition of causal factorization, 
up to some undetermined  phase factor.

\medskip
It also follows from equation \eqref{e.5.15}, cf.\ also 
\eqref{e.4.14} and \eqref{e.5.3},  that for 
any given $N,Q \in \Qzw$ 
the operators $S_N(Q) \doteq S(N)^{-1} S(Q + N)$ commute with all
perturbed Weyl operators
$W_N(f) = W(K \Delta_A^N f)$ for test functions
$f$ having their support in the spacelike
complement of $\text{supp} \, Q$ with regard to the
metric~$g_N$. Note that under these circumstances 
$Q \Delta_A^N f = 0$ and $\Delta_R^{N+Q} f = \Delta_R^N f$,
hence $T_Q^N$ acts like the identity on
$(K \Delta_A^N) f$.  
Thus, presuming that the perturbed Weyl operators satisfy
the condition of Haag duality \cite{LoRoVe},
the operators $S_N(Q)$ are elements
of the von Neumann algebra generated by
$W_N(f)$ for test functions $f$ having their support
in any causally closed region containing $\text{supp} \, Q$.  
Whence, pairs of operators $S_N(P), S_N(Q)$ commute 
if the functionals $P,Q \in \Qzw$ have spacelike
separated supports relative to the metric $g_N$,
denoted by~$P \underset{N}{\perp} Q$,

\medskip
Let us mention as an aside that Haag duality 
has been established by Araki~\cite{Ar} in case of non-interacting
scalar fields on Minkowski space, \ie $N = 0$.
Apparently, a fully satisfactory proof for perturbations
\mbox{$N \in \Qzw$} of this field has not
yet appeared in the literature. Yet there exist
unpublished results to that effect \cite{MoVe}, so we
take it for granted here.  

\medskip
We extend now the operators $S(P)$, $P \in \Qzw$,  to arbitrary 
perturbations $\biP \in \Pzw$. This is accomplished by  
observations made in the preceding section. Namely,
given any quadratic perturbation $P$, we put 
for arbitrary constants $c$ and linear functionals
$L_f = (F_f - (1/2) \langle f, \Delta_D f \rangle)$, 
compare equation~\eqref{e.4.11},
\begin{align} \label{e.5.18}
  S(c + L_f + P)
  & \doteq e^{i(c - (1/2)\langle f, \Delta_D^P f \rangle)} \, S(P) \, W(K \Delta_A^P f)
\nonumber \\ 
& = e^{i(c - (1/2)\langle f, \Delta_D^P f \rangle)} \, W(K \Delta_R^P f) \, S(P) \, .
\end{align}
The second equality follows from the adjoint action 
of $S(P)$ on Weyl operators, cf.~\eqref{e.5.3},  and 
$T_P \, K \Delta_A^P = K \Delta_R^P$.

\medskip
The extended operators
satisfy, for fixed $P \in \Qzw$, the
causal factorization relations. To give an example, the
preceding relations imply after some elementary computation
that, $f,g \in \Test(\Mc)$,
\be
S(F_f + P) S(P)^{-1} S(F_g + P) = e^{i \langle f, \Delta_A^P g \rangle} \,
S(F_f + F_g + P) \, .
\ee
Thus if $\text{supp} f \underset{P}{\succ} \text{supp} \, g$,
the phase factor is equal to $1$, in accordance with
the condition of causal
factorization. In a similar manner one verifies the causal factorization 
for all products of Weyl operators and the extended operators
involving a fixed quadratic perturbation. In other
words, the ambiguities in the phase factors 
appearing in the causal factorization relations of
the unitaries $S(\biP)$ depend only on the
quadratic parts \mbox{$P \in \Qzw$} of the functionals $\biP \in \Pzw$.

\medskip 
Relation \eqref{e.5.18} also implies that the extended operators
satisfy the dynamical condition, involving the Lagrangian $\Lag_0$. Since
constant functionals factor out from this condition, it suffices
to verify this assertion for functionals of the form $(F_f^P + P)$ for
arbitrary $f \in \Test(\Mc)$. A by now
routine computation shows that for perturbations $P \in \Qzw$
one obtains for the shifted functionals the equality  
\be
(F_f^P + P)^{\phi_0} + \delta \Lag_0 (\phi_0)
= F^P_{f + (K + P) \phi_0} + P \, , \quad 
\phi_0 \in \Ec_0 \, .
\ee
Thus 
\begin{align}
& S(P)^{-1}S((F_f^P + P)^{\phi_0} + \delta \Lag_0 (\phi_0)) =
S(P)^{-1}S(F^P_{f + (K + P) \phi_0} + P) \nonumber \\
& = W_P(f + (K + P) \phi_0) = W_P(f) = S(P)^{-1} S(F^P_f + P)  \, ,
\end{align}
where in the second equality we made use of the definition
\eqref{e.4.8} of the perturbed Weyl operators. The 
third equality is a consequence of the Weyl relations and the
fact that $W_P((K + P)\phi_0) = 1$. 
So we arrive, as claimed, at 
\be \label{e.5.22}
S(\biP^{\phi_0} + \delta \Lag_0(\phi_0)) = S(\biP) \quad \text{for} \quad
\biP \in \Pzw \, , \ \phi_0 \in \Ec_0 \, .
\ee

We summarize the results obtained in this section in a proposition.

\begin{proposition} \label{p.5.1}
  Let $\biP \in \Pzw$. There exist unitary operators $S(\biP)$
  on Fock space, inducing automorphisms of the Weyl algebra, which are 
  determined by equation~\eqref{e.5.18}. These operators
  satisfy the dynamical equation  
  \be \label{e.5.23}
  S(\biP^{\phi_0} + \delta \Lag_0(\phi_0)) =
  S(\biP) \, , \quad\ \phi_0 \in \Ec_0  \, .
  \ee
  Moreover, for any admissible triple of
  functionals $\biP,\biQ,\biN \in \Pzw$ satisfying
  $\biP \underset{\biN}{\succ} \biQ$,
  there exists a phase
  $\alpha(N|P,Q) \in \TT \doteq \{\xi \in \CC : | \xi | = 1\}$, 
  depending only on
  the quadratic parts $P,Q,N$ of the functionals, such that 
  \be \label{e.5.24}
  S(\biP + \biN) S(\biN)^{-1} S(\biQ + \biN)
  = \alpha(N|P,Q) \, S(\biP + \biQ + \biN) \, .
  \ee
  If $\biP, \biQ$ are spacelike separated,
  $\biP \underset{\biN}{\perp} \biQ$, the product in \eqref{e.5.24}
  is symmetric
  in  $\biP$, $\biQ$, \ie  $\alpha(N|P,Q) = \alpha(N|Q,P)$. 
\end{proposition}  

The family of functionals $\Pc$ is stable under translations,
yet not under Lorentz transformations because of the choice 
of a time direction in our standing assumption. 
Since there exists a unitary representation 
$\lambda \mapsto U(\lambda)$ of the Poincar\'e group on Fock space,
one can proceed from the operators $S(P)$, $P \in \Pc$,
to operators which are compatible
with any other choice of the time direction.
Namely, given $\lambda$, the unitaries  
$U(\lambda) S(\biP) U(\lambda)^{-1}$ induce automorphisms of the Weyl
operators whose quadratic part is fixed by the transformed
single particle operators  
$\overline{T}_{P_\lambda} \doteq U(\lambda) \overline{T}_P U(\lambda)^{-1}$.
Thus these unitaries comply, for adequate 
$\lambda$, with any choice of the time direction
in the standing assumption and satisfy the proposition
as well. In particular, the phase $\alpha$ in the
proposition can be chosen to be Poincar\'e invariant. 

\section{Phase factors and causal factorization}

We turn now to the problem of fixing the phases of the operators
$S(\biP)$, $\biP \in \Pzw$, so that they fully comply with the
causal factorization condition for the restricted
set of functionals. A somewhat simpler problem  
was treated by Scharf and Wreszinski~\cite{Scharf-Wreszinski}
for the case of a Fermi field, coupled to an external electromagnetic
field, cf.\ also \cite{GB}.                          
The kinetic perturbations are more singular, however, and 
an analogous computational approach, based on explicit expressions
for the factors $\alpha$ in \eqref{e.5.24}, cf.\ for
example \cite{Rob},  would require
some coherent non-perturbative renormalization scheme.\footnote{This 
is related to the problem of associating determinants to hyperbolic
  differential operators. For recent progress in the case of elliptic
  operators see \cite{Dang} where, however, the class of allowed
  perturbations is less singular.} 

\medskip
We therefore adopt here a different strategy. Relying on the 
results of Wald~\cite{Wald}, we have established in
the preceding section the existence of unitary operators
$S(P)$ on Fock space, which determine a projective representation of 
the group $\fQ$, generated by the operators $T_P$ for quadratic 
perturbations $P \in \Qzw$ on the single particle
space. The cohomology of this representation 
is known to be non-trivial due to the appearance of Schwinger
terms, cf.\ \cite{La} and references quoted there. Yet
these singularities are expected not to  
affect the causal factorization, involving 
perturbations with disjoint supports. We therefore focus on 
the projective causal factorization equation, stated in
Proposition \ref{p.5.1}, and look at it from a cohomological
point of view.

\medskip
Let $\alpha(N|P,Q) \in \TT$ be
the phase factors appearing in 
equation \eqref{e.5.24} for
quadratic functionals $P,Q,N \in \Qzw$. 
We begin by exhibiting two basic relations
satisfied by them,
which are used time and again. They
are a consequence of the associativity of
the underlying operator products.
We say that $\alpha(N|P,Q)$ is \textit{well defined}
if $P,Q,N \in \Qzw$ is an
admissible triple, satisfying the
causality condition $P \underset{N}{\succ} Q$.
\begin{lemma} \label{l.6.1} 
  Let $P_1, P_2, Q_1, Q_2, N \in \Qzw$. With 
  $P \doteq P_1 + P_2$ and  \mbox{$Q \doteq Q_1 + Q_2$}, one has
  \begin{align} \label{e.6.1}
    \alpha(N|P,Q) & =  \alpha(N|P_1,Q)  \
    \alpha(N + P_1|P_2,Q)  \nonumber  \\
    & = \alpha(N|P, Q_1) \ \alpha(N + Q_1|P, Q_2)  \, , 
  \end{align}
  provided all phases $\alpha$ are well defined.
\end{lemma}
\noindent \textbf{Remark:} 
These relations comprise within the present context the
essential part of the information contained in the cocycle equations, 
determined by the underlying projective representation of $\fQ$.
\begin{proof}
We have 
\begin{align} 
&  \alpha(N|P_1,Q) \, \alpha(N+P_1|P_2,Q) \ S(P + N + Q) \nonumber \\
  &  =  \alpha(N|P_1,Q) \, \alpha(N+P_1|P_2,Q) \
  S(P_2 + (P_1 + N) + Q) \nonumber \\
&  = \alpha(N|P_1, Q)) \ 
  S(P_2 + (P_1 + N)) S(P_1 + N)^{-1} S(P_1 + N + Q )  \\
&  = S(P + N) \underbrace{S(P_1 + N)^{-1} S(P_1 + N)}_{=1}
  S(N)^{-1} S(N + Q)  \nonumber \\
& = \alpha(N|P,Q) \ S(P + N + Q)  \, . \nonumber 
\end{align}
So the first equality in the statement follows.  
The second equality is obtained in a similar manner.  \qed 
\end{proof}

It is our goal to show that there exists a collection
of phases $\beta(P) \in \TT$, $P \in \Qzw$, 
such that for any admissible triple of functionals 
$P,Q,N \in \Qzw$ with $P \underset{N}{\succ} Q$, one has
\be \label{e.6.3} 
\alpha(N|P,Q) = \beta(P+N)^{-1} \beta(N) \beta(Q + N)^{-1}
\beta(P + Q + N) \, .
\ee
Note that for any choice of phases $\beta$, the
expression on the right hand side satisfies the
equalities in the preceding lemma. So, in other words,
we want to prove that these
equalities admit only such trivial solutions, akin to
the coboundaries solving cocycle equations
in cohomology theory. 
Multiplying each operator $S(\biP)$, $\biP \in \Pzw$, with the
phase factor $\beta(P)$, corresponding to the quadratic
part $P$ of $\biP$, the resulting operators satisfy the
proper causal factorization relation \eqref{e.5.24}, where
the phase factor $\alpha$ is identical to $1$. Moreover, since
the quadratic part $P$ of $\biP$ is not affected in the
dynamical relation \eqref{e.5.23}, this relation still holds true for the
modified operators $\beta(P) S(\biP)$, $\biP \in \Pzw$.
We thereby arrive at the main result of this article.

\begin{theorem} \label{t.6.2}
  Let $\Alg_2$ be the dynamical C*-algebra generated
  by unitaries $S(\biP)$, $\biP \in \Pzw$, which satisfy the dynamical
  condition (i) for the Lagrangian $\Lag_0$ of a scalar field with 
  mass $m \geq 0$ in $d > 2$ spacetime dimensions,
  as well as the causal
  factorization equation (ii). If $m > 0$, this algebra is
  represented by an extension of the
  Weyl algebra on the (positive energy) Fock space
  for any value of $d$; if
  $m = 0$, the dimension must satisfy $d \geq 4$.
\end{theorem}  

\medskip
Since the proof of relation \ref{e.6.3} is 
cumbersome, consisting of several steps, 
we begin with an outline of our argument. The functionals
$P \in \Qzw$, involving symmetric tensors and scalars,
depend on test functions $p$
on $\RR^d$, having values in a real vector space of
dimension $n(d) = d(d + 1)/2 + 1$. Our standing assumption
restricts these values to a 
convex set $\Kc \subset \RR^{n(d)}$ which is contractible,
\ie it is mapped into itself by scaling it  
with factors less than~$1$.
This set can be covered by an increasing net
of compact, convex and contractible 
subsets $\Kcc \subset \Kc$, $c \geq 1$, related to  
metrics of Minkowski type,
$\eta^{\, c}(x,x) = c^2 x_0^2 - \bix^2$, $x \in \RR^d$. 
The metric $\eta^{\, c}$ dominates all
metrics $g_P$ where $p$ takes values in $\Kcc$, \ie 
 the light cone
fixed by~$\eta^{\, c}$ contains the lightcone determined
by the metric~$g_P$. (See the appendix.)
The subset of functionals in $\Qzw$
involving test functions with values in
$\Kcc$ is denoted by $\Qzw(\Kcc)$.

\medskip 
In our analysis of the phases $\alpha(N|P,Q)$, we need to
consider limited numbers of (at most six)  
admissible triples of functionals \mbox{$P,Q,N \in \Qzw$}.
Any such collection of functionals is, together with the respective sums,
contained in some $\Qzw(\Kcc)$ for sufficiently large $c$. 
Making use of this fact, we can simplify the discussion of the causal
relations between the functionals appearing in the phases.  

\medskip 
Given any $c \geq 1$ and admissible triples
$P,Q,N \in \Qzw(\Kcc)$ satisfying  the causality condition 
$P \underset{N}{\succ} Q$, we restrict the corresponding  
phases $\alpha(N|P,Q)$ to the subset of triples  
satisfying the stronger causality condition
$P \csucc Q$. The latter symbol indicates that
the functional $P$ does not intersect the past of
the functional $Q$ with regard to the metric $\eta^{\, c}$,
\ie $\mbox{supp} \, P \cap \Jm{(\mbox{supp} \, Q)} = \emptyset$
in an obvious notation. Thereby, the causal relations between
the restricted functionals in $\Qzw(\Kcc)$ can be
discussed in a simpler, unified manner.
In order to mark this step, we denote 
the restricted phases by $\ualpha(N|P,Q)$ and introduce
the following terminology. 

\medskip \noindent
\textbf{Definition:} \ Let $c \geq 1$. A finite collection of 
phases $\ualpha(N_i|P_i, Q_i)$ for given admissible triples
$P_i, Q_i, N_i \in \Qzw(\Kcc)$ 
is said to be \textit{well defined} 
if \mbox{$P_i \csucc Q_i$} for $\, i = 1, \dots , N$.
In particular, the equalities \eqref{e.6.1} are satisfied
by such well defined collections of restricted phases.  

\medskip 
A major part of our argument consists of the proof
that the restricted phases $\ualpha(N|P,Q)$ can be extended
to a considerably larger set of functionals. 
As we will see, they have unique extensions $\oalpha(N|P,Q)$, being 
defined for admissible triples $P,Q,N \in \Qc(\Kcc)$  
with $\text{supp} \, P \cap \text{supp} \, Q = \emptyset$.
We will then show that these extensions are the restrictions
to $\Qc(\Kcc)$ of a global phase $\oalphao(N|P,Q)$ which is defined
for all admissible triples $P,Q,N \in \Qc$ satisfying
$\text{supp} \, P \cap \text{supp} \, Q = \emptyset$.
Moreover, $\oalphao$ coincides with the original phase
$\alpha$ on its domain.
The more transparent properties of $\oalphao$ will enable us
to prove that there exist phase factors $\beta(P) \in \TT$,
$P \in \Qc$, which trivialize it. That is, 
equation~\eqref{e.6.3} is satisfied for all 
admissible triples $P,Q,N \in \Qc$ with 
$\text{supp} \, P \cap \text{supp} \, Q = \emptyset$.
Thus,  \textit{a fortiori}, $\alpha$ can be trivialized.

\medskip
We turn now to the proof that the restricted phases
$\ualpha(N|P,Q)$ can be extended, as indicated. 
There we make use of the fact that
the phases are symmetric for spacelike separated
$P,Q$, cf.\ Proposition~\ref{p.5.1}. In accordance with our 
conventions, we will only consider pairs of
functionals which are spacelike separated with regard
to the metric $\eta^{\, c}$. It is note worthy 
that the condition of Haag duality, entailing the
symmetry of the phases, then follows already
from the seminal results of Araki \cite{Ar}.
A crucial step towards the extension of the phases is the following
lemma.

\begin{lemma} \label{l.6.3}
  Let $P_1, P_2, Q, N \in \Qzw(\Kcc)$. Then
\be \nonumber \label{alpha1} 
  \ualpha(N + P_1| Q, P_2) \
  \ualpha(N|P_1 ,Q) = \ualpha(N + P_2| P_1,Q) \ \ualpha(N| Q, P_2)\ .
\ee
if all occurring phases $\ualpha$ are well defined, 
cf.\ the preceding definition. 
\end{lemma}

\medskip \noindent
\textbf{Technical remark:}  
In the proof of this lemma, as well as in subsequent arguments, 
  we will make
  use of the fact that any functional $P \in \Qzw(\Kcc)$
  can be split within $\Qzw(\Kcc)$ into
  ``locally convex'' combinations 
  of functionals. This is accomplished by multiplying the
  (tensor-valued) test function $p$, underlying $P$, with some
  ``pointwise convex'' partition of  
  unity, $p_k \doteq \chi_k p$, 
  where $0 \leq \chi_k \leq 1$ are smooth functions
  and $\sum_{k = 1}^n \chi_k = 1$ on the support of $p$.
    Since $\Kcc$ is convex, the functionals
  $P_k$, which are obtained by replacing $p$ in $P$ by
  $p_k$, are contained in $\Qzw(\Kcc)$,
  $k = 1, \dots , n$. By some slight abuse of notation, we 
  put $P_k \doteq \chi_k P$, giving  
  \mbox{$\sum_{k = 1}^n P_k = (\sum_{k=1}^n \chi_k) P = P$}. Choosing suitable  
  pointwise convex partitions, we will split in this manner given
  functionals $P$
  into combinations of functionals with prescribed support properties,
  determined by the supports of $\chi_k$. For the sake of
  shortness, we omit the phrase 
  ``pointwise convex'' in the following. We will also use the  
  notation $\Jcap{ } \doteq \Jp{ } \cap \Jm{ }$ and
  $\Jcup{ } \doteq \Jp{ } \cup \Jm{ }$.

\begin{proof}
  For the proof of the lemma, we proceed from the underlying 
  condition $P_1 \csucc Q \csucc P_2$. So there exists
  a decomposition $P_1 = P_+ + P_0$ such that 
  $\supp P_+ \cap \Jm{(\supp Q \cup \supp P_2)} = \emptyset$
  and  $\supp P_0 \, \cap \, \Jcup{(\supp Q)} = \emptyset$.
  Making use Lemma \ref{l.6.1}, we then split the phases
  appearing in the statement: our underlying strategy
  consists of moving, whenever possible, $P_0$ to the first
  entry and $P_+$, $P_2$ to the second, respectively third, 
  entry. For the first factor, appearing on the left hand
  side of the equality in the statement, we obtain
\begin{align}
& \alpha(N + P_+ + P_0|Q, P_2)
=\alpha(N + P_0|Q + P_+, P_2) \, \alpha(N +  P_0|P_+, P_2)^{-1} \nonumber \\ 
& =\alpha(N + P_0 + Q|P_+, P_2) \, 
\alpha(N + P_0|Q, P_2) \, \alpha(N + P_0|P_+ , P_2)^{-1} \, . 
\end{align}
For the second factor, we get
\be
\alpha(N|P_+ + P_0,Q) = \alpha(N + P_0|P_+,Q) \, \alpha(N|P_0,Q) \, . 
\ee
The factors appearing on the right hand side
of the equality are treated similarly. The first factor
yields
\begin{align}
&  \alpha(N + P_2|P_+ + P_0,Q) \, \alpha(N + P_2|P_0, Q)^{-1} =
  \alpha(N + P_2 + P_0|P_+,Q) \nonumber \\
  &  =\alpha(N + P_0|P_+, Q + P_2) \, \alpha(N + P_0|P_+, P_2)^{-1}
  \nonumber \\
  & = \alpha(N + P_0 + Q|P_+, P_2) \, 
  \alpha(N + P_0|P_+, Q) \, \alpha(N + P_0|P_+, P_2)^{-1} \ .
\end{align} 
The second factor gives  
\begin{align}
  & \alpha(N|Q,P_2) \, \alpha(N + P_2|P_0, Q)
  = \alpha(N|Q, P_0 + P_2) \nonumber \\
& = \alpha(N + P_0|Q, P_2)\alpha(N|Q, P_0)\ .
\end{align} 
Noticing that $\alpha(N|Q, P_0) = \alpha(N|P_0, Q)$ since
$\supp P_0 \perpc \supp Q$,
we conclude that the products of the 
phase factors on the left and right hand side of
the equality in the statement coincide, as claimed. 

\medskip
Note that the conditions on the entries of
the phase factors are met 
in each of the preceding steps; because the functionals
appearing there, 
as well as their respective sums, are convex combinations of
(sums of) the given functionals, and $\Kcc$ is convex and
contractible.  \qed 
\end{proof}
                                              
We are now in a position to extend the    
restricted phases $\ualpha$ to more general entries. This 
is accomplished in several steps. Let 
$P,Q,N \in \Qzw(\Kcc)$ be admissible and let  
\be \label{e.6.9}                                 
\supp  P \cap \Jcap{(\supp Q)} = \emptyset \, .
\ee
There exists a partition 
$\chi_+, \chi_-$ such that $\chi_+ + \chi_- = 1$
on the support of $P$ and 
$P_\pm \doteq \chi_\pm P$ satisfy 
$\supp P_\pm \cap \Jmp{(\supp Q)} = \emptyset$. Moreover, 
$N + P_\pm$
are locally convex combinations of  $N$ and $N+P$.
With these constraints on $P,Q$, we can define 
\be \label{e.6.10}
\oalpha(N|P,Q) \doteq  \ualpha(N|P_+,Q) \, \ualpha(N + P_+|Q, P_-) \, .
\ee
This definition amounts to a symmetrization with 
regard to the causal order of $P,Q$, \viz it implies 
$\oalpha(N|P,Q) = \oalpha(N|Q,P)$ if $P \csucc Q$ or $Q \csucc P$.
As we shall see, this relation holds for arbitrary functionals
with disjoint supports. But we first
need to verify that $\oalpha$, so defined, (i) extends $\ualpha$ and
(ii) does not depend on the split $P = P_+ + P_-$ within the above limitations. 

\medskip   
As to (i), we note that if $P \csucc Q$, then $P_-$ and $Q$ have
spacelike separated supports, $P_- \perpc Q$, 
and we may interchange these functionals in the second factor of the
right hand side of the preceding equality. It 
then follows from Lemma \ref{l.6.1} that
$\oalpha(N|P,Q) = \ualpha(N|P,Q)$. We also note that
according to Lemma~\ref{l.6.3}, one 
may interchange the role of $P_+$ and~$P_-$ in the definition. 

\medskip 
Concerning (ii), we remark that the ambiguities involved in the splitting 
of $P$ pertain only to the spacelike complement of the support of $Q$. So let
$ P = (P_+ + P_0) + (P_- - P_0)$
be another convex splitting, where
$P_0 \perpc Q$. Then, bearing in mind the symmetry of the phases in
$P_0,Q$, we have  \!
\begin{align}
  & \quad \ualpha(N|P_+ \! + \! P_0,Q)  \,
  \ualpha(N \! + P_+\!  + P_0|Q, P_- \! - \! P_0)  \\ 
& = \ualpha(N|P_+, Q) \,
  \underbrace{\ualpha(N \! + \! P_+|P_0, Q) \,
    \ualpha(N \! + \! P_+|Q, P_0)^{-1}}_{=1}
  \, \ualpha(N \! + \! P_+|Q, P_-)  \nonumber \, ,
\end{align} 
proving that the extension $\oalpha$ is well defined. The extended phase
satisfies cocycle relations analogous to those established
for $\alpha$ in Lemma~\eqref{l.6.1}.

\begin{lemma} \label{l.6.4}
  Let $P_1, P_2, Q, Q_1, Q_2, N \in \Qzw(\Kcc)$. Putting
  $P = P_1 + P_2$, one has 
\begin{align} \label{e.6.12}
& \qquad \oalpha(N|P_1 \! + \! P_2, Q)  =
\oalpha(N|P_1, Q) \ \oalpha(N \! + \! P_1|P_2, Q)  \\[1mm] 
\label{e.6.13}
& \oalpha(N|P, Q_1) \ \oalpha(N \! + \! Q_1|P, Q_2) =
\oalpha(N|P, Q_2) \ \oalpha(N \! + \! Q_2|P, Q_1) \, , 
\end{align}
provided all terms are well defined. The latter condition
now implies that all phase factors contain 
admissible triples in $\Qzw(\Kcc)$,
where the functionals in their second and third
entry have disjoint supports, in agreement with condition~~\eqref{e.6.9}. 
\end{lemma}
\begin{proof}
  For the proof of the first equality in  \eqref{e.6.12}, we
  split $P_i = P_{i +} + P_{i -}$, 
  where $P_{i \pm}$ are functionals, $i = 1,2$,
  with appropriate support properties
  relative to $Q$, in accordance with 
  definition \eqref{e.6.10} of the extended phases. 
  The left hand side of \eqref{e.6.12} is then defined by 
\be 
\ualpha(N|P_{1 +} + P_{2 +}, Q) \
\ualpha(N + P_{1 +} + P_{2 +}|Q, P_{1-} + P_{2-}) \, .
\ee
Applying Lemma \ref{l.6.1} to every factor, we obtain
\begin{align} \label{e.6.15}
& \ualpha(N|P_{1 +} + P_{2 +}, Q)  \ = \
\ualpha(N|P_{1 +}, Q) \ \ualpha(N + P_{1 +}|P_{2 +}, Q) \, , \nonumber \\
& \ualpha(N + P_{1 +} + P_{2 +}|Q, P_{1-} + P_{2-})   
= \ualpha(N + P_{1 +} + P_{2 +}|Q, P_{1 -})  \nonumber  \\
& \cdot \ualpha(N + P_{1 +} + P_{2+} + P_{1 -}|Q,P_{2 -}) \, .  
\end{align}
The factors appearing on the right hand side of \eqref{e.6.12}
are given by 
\begin{align} \label{e.6.16}
  & \oalpha(N|P_1, Q) = \ualpha(N|P_{1 +}, Q) \
  \ualpha(N + P_{1 +}|Q, P_{1 -}) \, ,  \\
  & \oalpha(N + P_1|P_2, Q) = \ualpha(N + P_1|P_{2 +}, Q)
  \ \ualpha(N + P_1 + P_{2 +}|Q, P_{2 -}) \, . \nonumber
\end{align}  
Comparing the four factors appearing on the right hand sides of the
equalities in \eqref{e.6.15}, respectively \eqref{e.6.16}, we see that
two of them coincide. For the products of the remaining pairs, 
we get 
\begin{align}
&  \ualpha(N + P_{1 +}|P_{2 +}, Q) \ \ualpha(N + P_{1 +} + P_{2 +}|Q, P_{1 -})
  \nonumber \\
  &  = \oalpha(N + P_{1 +}| P_{1 -} + P_{2 +}, Q)  \nonumber \\
  & = \ualpha(N + P_{1 +}|Q, P_{1 -}) \
  \ualpha(N + P_{1 +} + P_{1 -}|P_{2 +}, Q) \, , 
\end{align}  
completing the proof of relation \eqref{e.6.12}. 

\medskip
Turning to the proof of relation \eqref{e.6.13}, we make use of  
the underlying condition 
$\supp P
\cap \big( \Jcap{(\supp Q_1)} \cup \Jcap{(\supp Q_2)} \big) = \emptyset$.
So there exists a convex decomposition
$P =
P_{\tp \tp} + P_{\tp \tm} + P_{\tm \tp} + P_{\tm \tm}$
whose components satisfy \
$\supp P_{\sigma \sigma'} \cap \big(J^{\, c}_{-\sigma}(\supp Q_1) \cup
J^{\, c}_{-\sigma'}(\supp Q_2) \big) = \emptyset$ \ for $\sigma, \sigma' = \pm$. 
We then apply relation~\eqref{e.6.12}
to the phases appearing on the left hand side of equation~\eqref{e.6.13} 
and obtain 
\begin{align} \label{a3b}
&  \oalpha(N|P, Q_1) \\ 
  & = \oalpha(N|P_{\tp \tp} \! + \! P_{\tm \tm}, Q_1)  \
  \oalpha(N \! + \!  P_{\tp \tp} \! + \! P_{\tm \tm}|P_{\tp \tm} \! +
  \!  P_{\tm \tp}, Q_1)
  \ \nonumber \, . \\
& \oalpha(N \! + \! Q_1|P,Q_2) \label{a3a}  \\ 
  & = \oalpha(N \! + \! Q_1|P_{\tp \tp} \! + \! P_{\tm \tm}, Q_2) \
  \oalpha(N \! + \! Q_1 \! + \! P_{\tp \tp} \! + \! P_{\tm \tm}|
  P_{\tp \tm} \! + \! P_{\tm \tp}, Q_2)
  \ \nonumber \, .
\end{align}
The first factors on the right hand side of
equations \eqref{a3b}, respectively~\eqref{a3a},
are by definition equal to 
\begin{align} \label{a3d}
& \oalpha(N|P_{\tp \tp} \! + \! P_{\tm \tm}, Q_1) = 
\ualpha(N|P_{\tp \tp}, Q_1) \ \ualpha(N + P_{\tp \tp}|Q_1, P_{\tm \tm}) \, , \\
\label{a3c} 
& \oalpha(N \! + \! Q_1|P_{\tp \tp} \! + \! P_{\tm \tm}, Q_2) \nonumber \\
& \qquad \qquad \ =
\ualpha(N \! + \! Q_1|P_{\tp \tp}, Q_2) \
\ualpha(N \! + \! Q_1 \! + \! P_{\tp \tp}|Q_2, P_{\tm \tm}) \, .
\end{align}
Hence, applying Lemma \ref{l.6.1} twice, their product is given by 
\be \label{a3e}
    \ualpha(N|P_{\tp \tp}, Q_1 \! + \! Q_2) \
    \ualpha(N \! + \! P_{\tp \tp}|Q_1 \! + \! Q_2,P_{\tm \tm}) \, . 
\ee
It is thus symmetric in $Q_1$, $Q_2$.

\medskip 
The second factors on the right hand side
of \eqref{a3b} and \eqref{a3a} are treated 
similarly. There we have
\begin{align} 
& \label{a3g}
  \oalpha(N \! + \! P_{\tp \tp} \! + \! P_{\tm \tm}|P_{\tp \tm} \! + \!
  P_{\tm \tp}, Q_1) \\
  &  = \ualpha(N \! + \! P_{\tp \tp} \! + \! P_{\tm \tm}|Q_1, P_{\tm \tp}) \,
  \ualpha(N \! + \! P_{\tp \tp} \! + \! P_{\tm \tm} \! + \! P_{\tm \tp}|
  P_{\tp \tm}, Q_1) 
    \, , \nonumber  \\ \!
  \label{a3f}
  & \oalpha(N \! + \! Q_1 \! + \! P_{\tp \tp} \! + \! P_{\tm \tm}|P_{\tp \tm} \! + \!
  P_{\tm \tp}, Q_2)
  \\ 
  & = \ualpha(N \! + \! Q_1 \! + \! P_{\tp \tp} \! + \! P_{\tm \tm}|P_{\tm \tp}, Q_2)
  \, \ualpha(N \! + \! Q_1 \! + \! P_{\tp \tp} \! + \! P_{\tm \tm}
  \! + \! P_{\tm \tp}|Q_2, P_{\tp \tm}) 
  \, , \nonumber
\end{align}
We apply Lemma \ref{l.6.3} to 
the product of the first factors on the right hand side of 
\eqref{a3g}, \eqref{a3f}, changing the places of
$Q_1, Q_2, P_{\tm \tp}$ with the result 
\be \label{a3i}
\ualpha(N \! + \! Q_2 \! + \! P_{\tp \tp} \! + \!
P_{\tm \tm}|Q_1, P_{\tm \tp}) \
\ualpha(N \! + \! P_{\tp \tp} \! + \! P_{\tm \tm}|P_{\tm \tp}, Q_2) \, .
\ee
For the product of the second factors we obtain  
\be
\label{a3h}
\ualpha(N \! + \! Q_2 \! + \! P_{\tp \tp} \! +
\! P_{\tm \tm} \! + \! P_{\tm \tp}|P_{\tp \tm}, Q_1) \ 
\ualpha(N \! + \! P_{\tp \tp} \! + \! P_{\tm \tm} \! + \!
P_{\tm \tp}|Q_2, P_{\tp \tm}) \, .
\ee
Now the product of the first factors in \eqref{a3i} and \eqref{a3h} 
coincides by definition with the extended phase 
\be \label{e.6.27}
\oalpha(N \! + \! Q_2 \! + \! P_{\tp \tp} \! + \! P_{\tm \tm}|
P_{\tp \tm} \! + \! P_{\tm \tp}, Q_1) \, ,
\ee
and the product of the second factors in \eqref{a3i} and \eqref{a3h} yields 
\be \label{e.6.28} 
\oalpha(N \! + \! P_{\tp \tp} \! + \! P_{\tm \tm}|
P_{\tp \tm} \! + \! P_{\tm \tp}, Q_2)\ .
\ee
Since the product of \eqref{e.6.27} and \eqref{e.6.28}
coincides with the product of the second factors in
\eqref{a3b} and \eqref{a3a}, we conclude that this product
is also symmetric in $Q_1$, $Q_2$.
Noting once again that the phase factors,  
which appeared in intermediate steps,
were well defined for the respective triples in 
$\Qzw(\Kcc)$,
the proof of equality \eqref{e.6.13} is complete. \qed
\end{proof}

In a final step, we extend
$\oalpha(N|P,Q)$ to triples $P,Q,N \in \Qzw(\Kcc)$, where  
$P,Q$ have arbitrary disjoint supports,
\viz we also admit functionals $Q$ whose
support is not causally closed. 
Let $\Nc_1, \dots, \Nc_n$ be an open covering of $\supp Q$ such that
$\Jcap{(\Nc_i)} \, \cap \, \supp P = \emptyset$, $i= 1, \dots , n$.
Picking a corresponding partition of unity 
by test functions $\chi_i$, 
we obtain a decomposition $Q = Q_1 + \cdots + Q_n$ 
with $Q_i \doteq \chi_i \, Q$, $i = 1, \dots ,n$. We then put  
\begin{align} \label{e.6.29}
  & \oalpha(N|P,Q)  \\
  & \doteq \oalpha(N|P, Q_1) \ \oalpha(N \! + \! Q_1|P, Q_2)
  \ \oalpha(N \! + \! Q_1 \! + \! \cdots \! + \! Q_{n-1}| P, Q_n) 
  \nonumber \, .
\end{align}  

It follows from relation \eqref{e.6.13} that the right hand side of this
equality does not change if one exchanges the positions of
$Q_i$ and $Q_{i + 1}$. Hence it is stable
under arbitrary permutations of the $Q_i$, $i = 1, \dots , n-1$. 
It also does not depend on the chosen partition of unity,
as we will show next.

\medskip 
Let $\rho_i$, $i = 1, \dots , n$,  be another 
partition of unity for the chosen covering. We first consider the
cases where $\rho_i + \rho_j = \chi_i + \chi_j$ for some
pair {$i \neq j$} and all other test functions
coincide, $\rho_k = \chi_k$, $k \neq i,j$. 
According to the preceding observation, we may
reorder the indices and assume $i=1$, $j=2$.
Putting $R \doteq (\rho_1 - \chi_1) \, Q$, we obtain 
$(Q_1 + R) = \rho_2 \, Q$, \, $(Q_2 - R) = \rho_1 \, Q$.
Since $\supp P \, \cap  \, \supp \rho_i  Q = \emptyset $,
$i = 1,2$, we can apply relation~\eqref{e.6.12}, giving 
 \begin{align}
   & \oalpha(N|P,Q_1 \! +  \! R) \
   \oalpha(N  \!  + \!   Q_1  \! + \! R|P, Q_2 \! -  \!R) \\
 & = \oalpha(N|P, Q_1) \ 
 \underbrace{\oalpha(N \! + \! Q_1|P, R) \
   \oalpha(N \! + \! Q_1|P, R)^{-1}}_{=1} \
   \oalpha(N \! + \! Q_1|P, Q_2) \, . \nonumber
 \end{align} 
 We conclude that under these special changes of the partition of
 unity, the right hand side of definition \eqref{e.6.29}   
 does not change. But any other partition
 of unity can be reached in a finite number of steps from
 partitions of this special type, so the definition
 does not depend on it either. 

\medskip
Finally, the definition is also independent of the chosen covering. To
see this we proceed to refinements of the given covering and
corresponding refinements of the decompositions of the
functionals. Let, for example, $Q_1 = Q_{11} + Q_{12}$ be
such a refinement. Splitting \mbox{$P = P_+ + P_-$},
where $\supp P_\pm$ does not intersect $\Jmp{(\supp Q_1)}$,
respectively, we have 
\be
\oalpha(N|P,Q_1) = \ualpha(N|P_+, Q_1) \ \ualpha(N \! + \! P_+|Q_1, P_-) \, .
\ee
According to Lem\-ma~\ref{l.6.1}, the factors appearing on the right hand side
can be split into 
\begin{align}
&  \ualpha(N|P_+, Q_1) = \ualpha(N| P_+, Q_{11}) \
  \ualpha(N + Q_{11}|P_+, Q_{12}) \, , \\
&  \ualpha(N + P_+|Q_1, P_-) = \ualpha(N + P_+|Q_{11}, P_-) \
  \ualpha(N + Q_{11} + P_+|Q_{12}, P_-) \, . \nonumber 
\end{align}  
The product of the first factors on the right hand sides of these 
equalities gives $\oalpha(N|P, Q_{11})$ and that of the second
factors $\oalpha(N + Q_{11}|P, Q_{12})$. Thus we arrive at 
\be
\oalpha(N|P,Q_1) = \oalpha(N|P,Q_{11}) \ \oalpha(N + Q_{11}|P, Q_{12}) \, .
\ee
Iterating this argument, we see that definition \eqref{e.6.29} 
is invariant under finite refinements of the covering.
Since any two coverings have a joint refinement, it follows
that the extension of $\oalpha(N|P,Q)$ is well defined if $P, Q$ have disjoint
supports and all (sums of the) functionals are contained in~$\Qzw(\Kcc)$. 
The preceding results are used in the proof of the
following proposition. 

\begin{proposition} \label{p.6.5}
  The phase factors $\alpha$ appearing in Proposition \ref{p.5.1} can
  be extended to phases $\oalphao$ which are defined
  for admissible triples {$P,Q,N \in \Qc$}
  with $\supp P \bigcap \supp Q = \emptyset$ and satisfy 

\vspace*{-6mm}  
\begin{align} \label{e.6.33} 
\oalphao(N|P,Q) & = \oalphao(N|Q,P) \, , \\    
\label{e.6.34}
\oalphao(N|P_1 + P_2, Q) & = \oalphao(N|P_1, Q) \  
\oalphao(N + P_1|P_2, Q) \, .   
\end{align}

\vspace*{-1mm} \noindent 
These equalities uniquely fix this extension.
\end{proposition}
\begin{proof}
  Let $P,Q,N$ be any admissible triple with 
  $\supp P \bigcap \supp Q = \emptyset$. There exists some
  $c \geq 1$ such that $P,Q,N \in \Kcc$. As shown above,
  the restriction $\ualpha$ of $\alpha$ to admissible
  triples $P', Q', N' \in \Kcc$ satisfying
  $P' \csucc Q'$ can be 
  extended to
  phases $\oalpha$ which are defined on all admissible triples
  in $\Kcc$  with  $\supp P' \bigcap \supp Q' = \emptyset$.
  We shall see that this extension satisfies 
  the two equalities stated above.
  It is then clear that it is  unique because these equalities 
  comprise the defining equation for $\oalpha$ in terms of the
  restricted phases~$\ualpha$.
   
\medskip
In a first step we show                              
that the extended phase $\oalpha$ coincides
with the orig\-inal phase $\alpha$ on its full domain in  $\Qc(\Kcc)$. 
So let $P,Q,N \in \Kcc$ with 
$P \underset{N}{\succ} Q$. 
Any pair 
$(x,y) \in \supp P \times \supp Q$ satisfies either
$x \csucc y$, or $y \csucc x$. In the latter case, the point 
$x$ is spacelike separated from $y$ with regard to
the metric $g_N$ induced by $N$, $x \underset{N}{\perp} y$.
Thus, since the
supports of $P$, $Q$ are compact, we can split these
functionals with the help of suitable partitions
of unity into finite sums $P = \sum_i P_i$,
$Q = \sum_j Q_j$,  such that either
$P_i \csucc Q_j$, or $Q_j \csucc P_i$ and
$P_i \underset{N}{\perp} Q_j$. By repeated application 
of the basic Lemma \ref{l.6.1}, we obtain a  corresponding
decomposition of the phase $\alpha(N|P,Q)$, given by  
\be \label{e.6.54}
\alpha(N|P,Q) = \Pi_{i,j} \ \alpha(N + \sum_{k < i} P_k + \sum_{l < j} Q_l \, |
\, P_i, Q_j) \, .
\ee

The phases appearing on the right hand side
of this equality are
well defined, which  can be seen as follows: the causal structure
induced by their first entries coincides in the complement
of $\supp P \, \cup \, \supp Q$ with the causal structure fixed by $N$.
Thus, any future directed curve, emanating from a given point
in the support of $P$, will hit its boundary and
then propagate in positive timelike directions, fixed by the
metric $g_N$. Since $P \underset{N}{\succ} Q$, it will not reach the
past of $Q$. An analogous statement holds for past
directed curves emanating from points in the support
of $Q$. Hence,
putting $N_{i j} \doteq N + \sum_{k < i} P_k + \sum_{l < j} Q_l$,
one has $P_i \underset{N_{i j}}{\succ} Q_j$, as claimed.

Now factors in relation \eqref{e.6.54}, involving 
triples with $P_i \csucc Q_j$, coincide with the
restricted phase $\ualpha$ for these triples, whence with its 
extension $\oalpha$. If $Q_j \csucc P_i$, hence
\mbox{$P_i \underset{N}{\perp} Q_j$}, the entries in $\alpha$
can be interchanged because of the symmetry properties of $\alpha$
for functionals with spacelike separated supports.
So also in these cases the phase coincides with $\oalpha$ 
for the respective functionals. Since 
equality \eqref{e.6.54} holds also for the
extended phase, it follows that $\oalpha$
coincides with the original phase $\alpha$ 
on its domain.

\medskip
 Making use of relation \eqref{e.6.54}   
for the extended phase $\oalpha$ and noticing that the 
supports of the functionals $P_i, Q_j$ satisfy the
conditions stated after the defining equation \eqref{e.6.10},
it is apparent that the right hand side of
relation~\eqref{e.6.54} for $\oalpha$ is symmetric 
in $P,Q$. This proves equality \eqref{e.6.33} for $\oalpha$.
Since the first part of Lemma \ref{l.6.4} entails 
equality \eqref{e.6.34} for $\oalpha$, this establishes 
the two relations given in the statement
in case of $\oalpha$. 

\medskip
In the last step we show that for given admissible
triples $P,Q,N \in \Qc$, the extension of $\alpha$
does not depend on the value of $c$ chosen for the
embedding of the triple into $\Kcc$. So let
$\widehat{c} \geq c \geq 1$, hence $\Kc^{\widehat{c}} \supset \Kcc$. For pairs
$P,Q \in \Kcc$ the relation $P \overset{\widehat{c}}{\succ} Q$
implies $P \overset{{c}}{\succ} Q$. Hence
$\ualphao^{\widehat{c}}$ coincides with $\ualpha$ on
all admissible triples in $\Kcc$ 
satisfying the stronger causality condition. 
We proceed now as in the preceding step and decompose
$P = \sum_i P_i$ , $Q = \sum_j  Q_j$ such that for each pair
$(i,j)$ at least one of the relations
$P_i \overset{\widehat{c}}{\succ} Q_j$
or $Q_j \overset{\widehat{c}}{\succ} P_i$
holds. Adopting the notation in the preceding step,
we find in the first case 
\begin{align}
\oalpha(N_{ij}|P_i, Q_j) & = \ualpha(N_{ij}|P_i, Q_j) \nonumber \\
& = \ualphao^{\widehat{c}}(N_{ij}|P_i, Q_j) =
\oalphao^{\widehat{c}}(N_{ij}|P_i, Q_j) \, .
\end{align}
In the second case we obtain, 
bearing in mind the symmetry properties of the extended
phases in their second and third argument,
\begin{align}
\oalpha(N_{ij}|P_i, Q_j) & =  \oalpha(N_{ij}|Q_j, P_i) =
  \ualpha(N_{ij}|Q_j, P_i)  \\
& = \ualphao^{\widehat{c}}(N_{ij}|Q_j, P_i) =
  \oalphao^{\widehat{c}}(N_{ij}|Q_j, P_i) = 
  \oalphao^{\widehat{c}}(N_{ij}|P_i, Q_j) \, . \nonumber 
\end{align}
Thus, by another decomposition based on Lemma \ref{l.6.1}, we arrive at 
\begin{align}
  \oalpha(N|P,Q) & = \Pi_{i j} \, \oalpha(N_{i j}|P_i, Q_j) \nonumber \\
  & = \Pi_{i j} \, \oalphao^{\widehat{c}}(N_{i j}|P_i, Q_j)
  = \oalphao^{\widehat{c}}(N|P,Q) \, .
\end{align}    
So the phases $\oalpha$ are restrictions to $\Qc(\Kc^c)$, $c \geq 1$, 
of a phase $\overline{\alpha}$ which is defined on
all of $\Qc$, completing the proof. \qed
\end{proof}

We will show now that the extended
phases $\oalphao$ can be trivialized. Trivial solutions of the
equalities in the preceding proposition are obtained by
picking phases $\beta(P) \in \TT$, $P \in \Qc$,
and putting
\be \label{e.6.35}
\delta \beta(N|P,Q) \doteq 
\beta(P+N)^{-1} \beta(N) \beta(Q+N)^{-1} \beta(P + Q + N)
\, .
\ee
They correspond to coboundaries in cohomology theory. 
Thus we need to exhibit phase factors $\beta$ for which $\oalphao$
can be represented in this form. The construction of these
phase factors will be accomplished in successive steps. Namely, 
we will adjust the phases $\beta$ for increasing subsets of functionals
in $\Qzw$ 
such that the preceding equality is satisfied in each step
by $\oalphao$, restricted to the respective subsets of
functionals. The desired result is then obtained by some
limiting argument. 

\medskip
It will be convenient to describe this procedure by
an iterative scheme. To this end we multiply $\oalphao$ 
with the inverse of \eqref{e.6.35},
involving the phases determined in each step,
$\oalphao(N|P,Q) \mapsto \oalphao(N|P,Q) \, \delta \beta(N|P,Q)^{-1}$.
The resulting phases still satisfy both equations in the
preceding proposition and are equal to~$1$ on increasing
subsets of functionals. From the point of view of
cohomology theory, we are staying by this procedure
in the cohomology class of $\oalphao$. We therefore 
denote the phase factors, being modified in this
manner, again by $\oalphao$.

\medskip
Turning to the construction, let $\Reg_1,\Reg_2 \subset \Mc$
be disjoint compact regions and let
$\chi_0, \chi_1, \chi_2$  be a partition of unity such that
$\chi_1, \chi_2$ have disjoint supports, are
equal to $1$ on $\Reg_1$, respectively $\Reg_2$,
and \mbox{$\chi_0 = 1 - \chi_1 - \chi_2$}.
Let $P,Q,N \in \Qzw$ be any admissible triple
such that $\supp P \Subset \Reg_1$, $\supp Q \Subset \Reg_2$,
where the symbol~$\Subset$ indicates that the supports are
contained in the open interior of the given regions. 
Setting $N_j \doteq \chi_j N$, \mbox{$j=0,1,2$}, it follows from 
equations \eqref{e.6.34} and \eqref{e.6.33} that
\begin{align} \label{e.6.36}
\oalphao(N|P, Q) & =
\oalphao(N_0 \! +  \! N_2| P \!  +  \! N_1, Q) \
\oalphao(N_0  \! +   \! N_2| N_1, Q)^{-1}  \nonumber \\
\oalphao(N_0 \! +  \! N_2| P \!  +  \! N_1, Q) \
 & =
\oalphao(N_0|P  \! +  \! N_1, Q  \! +  \! N_2) \
\oalphao(N_0|P  \! +  \! N_1, N_2)^{-1}  \nonumber  \\
\quad \oalphao(N_0  \! +   \! N_2| N_1, Q)^{-1}
& = \oalphao(N_0|N_1, Q  \! +  \! N_2)^{-1} 
\ \oalphao(N_0|N_1, N_2) \, .
\end{align}
With this input, we define  
$\beta(R) \doteq \alpha(\chi_0 R|\chi_1 R, \chi_2 R)$
for $R \in \Qzw$. Making use of the fact that 
$\chi_0 P = \chi_2 P = 0$ and $\chi_0 Q = \chi_1 Q = 0$, 
the equalities \eqref{e.6.36} imply that
\be \label{e.6.37}
\oalphao(N|P,Q) = \beta(P + Q + N) \beta(P + N)^{-1}
\beta(Q + N)^{-1} \beta(N) 
\ee
for the restricted set of triples $P,Q,N$. Thus $\oalphao$ is trivial
for such triples. Multiplying
$\oalphao$ with the inverse of the right hand side, we
obtain an improved phase $\oalphao$ which still satisfies 
the equations in Proposition \ref{p.6.5} and, in addition, 
is equal to $1$ if $\supp P \Subset \Reg_1$,
$\supp Q \Subset \Reg_2$.

\medskip 
Given any $\oalphao$ with these properties, we repeat the
preceding procedure. So let $\Reg_3, \Reg_4 \subset \Mc$ be another pair
of disjoint compact regions 
and let $\chi'_0, \chi_3, \chi_4$ be a corresponding
partition of unity. As in the preceding step, we put 
\mbox{$\beta(R) \doteq \oalphao(\chi'_0 R|\chi_3 R, \chi_4 R)$} 
for $R \in \Qzw$.
Thus $\oalphao$ satisfies equation~\eqref{e.6.37}
for the respective triples. Multiplying it with
the inverse of the right hand side,  we obtain
a modified phase $\oalphao$ which satisfies
the equations in Proposition 6.5 and is
equal to $1$ if $\supp P \Subset \Reg_3$,
\mbox{$\supp Q \Subset \Reg_4$}. As a matter of fact, it turns
out that this modified phase is
still equal to $1$ also for the original triples $P,Q,N$
satisfying  $\supp P \Subset \Reg_1$,
$\supp Q \Subset \Reg_2$.

\medskip
Making use of the properties of the improved phases $\oalphao$,
established in the preceding step, and of Proposition \ref{p.6.5}, we have for
admissible $P,Q,R,N$ with $\supp P \Subset \Reg_1$,
$\supp Q \Subset \Reg_2$
\begin{align}  \label{e.6.38} 
  &  \oalphao(N + P|R,Q) 
   \, = \, \oalphao(N|P + R,Q) \ \oalphao(N|P,Q)^{-1} \nonumber \\
&  =  \, \oalphao(N|P + R,Q)  
   \, = \, \oalphao(N|R,Q) \ \oalphao(N + R|P,Q) \nonumber \\
&  = \, \oalphao(N| R,Q) \, . 
\end{align}
This equality will be used at several points in the 
proof of the following important lemma. 
\begin{lemma} \label{l.6.6}
  Let $\beta$ be the phases, determined in the preceding step 
  for disjoint
  regions $\Reg_3, \Reg_4$ from a given $\oalphao$, which is equal to $1$ 
  on pairs of functionals with support in
  disjoint regions $\Reg_1, \Reg_2$. Then
  \be \label{e.6.39}
  \beta(P + Q + N) \beta(P + N)^{-1}
\beta(Q + N)^{-1} \beta(N) = 1
\ee
for admissible triples $P,Q,N \in \Qzw$ with
$\supp P \Subset \Reg_1$, $\supp Q \Subset \Reg_2$.
\end{lemma}
\begin{proof}
We put 
$R_0 = \chi'_0 R$, $R_3 = \chi_3 R$ , $R_4 = \chi_4 R$, thus
$R_0 + R_3 + R_4 = R$, $R \in \Qzw$, and consider
for admissible triples $P,Q,N$ the phase 
\be
\beta(P + Q + N) = 
\oalphao(N_0 + P_0 + Q_0|N_3 + P_3 + Q_3, N_4 + P_4 + Q_4) \, .
\ee
Making use of the given support properties of $P,Q$, we will
split this expression with the help of Proposition \ref{p.6.5}
into a product of phases, where $P,Q$ do not
appear, both, in the same factor. 
This is a somewhat lengthy procedure. We begin by 
applying relation~\eqref{e.6.34} twice, giving 
\be
\oalphao(N_0 + P_0 + Q_0|N_3 + P_3 + Q_3, N_4 + P_4 + Q_4) 
= \alpha_1 \, \alpha_2 \, \alpha_3 \, \alpha_4 \, ,
\ee 
where 
\begin{align}
  \alpha_1 & = \oalphao(N_0 + N_3 + N_4 + P_0 + Q_0|P_3 + Q_3, P_4 + Q_4)\ ,
\nonumber \\
\alpha_2 & = \oalphao(N_0 + N_4 + P_0 + Q_0|N_3, P_4 + Q_4) \,
\nonumber \\
\alpha_3 & = \oalphao(N_0 + N_3 + P_0 + Q_0|P_3 + Q_3, N_4)\ ,
\nonumber \\
\alpha_4 & = \oalphao(N_0 + P_0 + Q_0|N_3, N_4) \, .
\end{align}
Turning to 
$\alpha_1$, we apply again relation \eqref{e.6.34} twice and obtain 
\begin{align}
  \alpha_1 = & \ \oalphao(N + P_0 + Q|P_3, P_4) \
  \oalphao(N + P_0 + Q_0 + Q_3|P_3, Q_4) \nonumber \\
& \! \! \cdot 
  \oalphao(N + P_0 + Q_0 + Q_4|Q_3, P_4) \ \oalphao(N + P_0 + Q_0|Q_3, Q_4) \, .
\end{align}
Since the second and third entries in the two 
middle factors have support in $\Reg_1$, respectively $\Reg_2$,
these factors are equal to~$1$. 
By relation \eqref{e.6.38} we can omit in the first factor $Q$ and
in the fourth factor $P_0$, hence 
\be
\alpha_1= \oalphao(N + P_0|P_3, P_4) \ \oalphao(N + Q_0|Q_3, Q_4) \, .
\ee
To the second factor $\alpha_2$ we apply both equalities 
Proposition \ref{p.6.5}, giving 
\be
\alpha_2 = \oalphao(N_0 + N_4 + P_0 + Q_0|N_3, P_4) \
\oalphao(N_0 + N_4 + P_0 + Q_0 + P_4|N_3, Q_4) \, .
\ee
According to relation \eqref{e.6.38} we can omit $Q_0$ in the first
factor and $P_0 + P_4$ in the second factor with the result
\be
\alpha_2 = \oalphao(N_0 + N_4 + P_0|N_3, P_4) \
\oalphao(N_0 + N_4 + Q_0|N_3, Q_4) \, .
\ee
The third factor $\alpha_3$ is treated similarly  and we find
\be
\alpha_3 = \oalphao(N_0 + N_3 + P_0|P_3, N_4) \
\oalphao(N_0 + N_3 + Q_0|Q_3, N_4) \, .
\ee

\medskip
We turn now to the factor $\alpha_4$. For its analysis, 
we need a finer resolution of the functional $N$.
To this end we choose a partition of unity 
$\rho_0 + \rho_1 + \rho_2 = 1$ such that $\supp\rho_1 \subset \Reg_1$,
$\supp\rho_2\subset \Reg_2$ and
$\supp \rho_0\cap(\supp P \cup \supp Q) = \emptyset$.
Since $P,Q$ have supports in the interior of the
respective regions, such a partition exists 
and the supports of $\rho_1 N$, $\rho_2 N$ are
contained in $\Reg_1$, respectively $\Reg_2$. 
We then consider the functionals
$N_i^j \doteq \rho_j \chi_i N$ for $i =3,4$ and $j = 0,1,2$.  
Decomposing $N_3, N_4$ in the second and third entry of
$\alpha_4$, we move the terms appearing in the
corresponding sums successively to the first entry
with the help of the two equalities in
Proposition \ref{p.6.5}. We thereby arrive at a
product of nine factors of the form
\be \label{e.6.48} 
\alpha_{j,k} \doteq \oalphao(P_0 + Q_0 + M_{j k}|N_3^j, N_4^k) \, ,
\quad j,k = 0,1,2 \, ,
\ee
where $M_{j k}$ is a sum of $N_0$ and certain specific terms in the
decomposition of $N_3$, $N_4$.
As a matter of fact, this assertion  
becomes more transparent by proceeding in reverse. 
Beginning with
$\oalphao(P_0 + Q_0 + N_0| N_3^0, N_4^0)$, one 
builds $\alpha_4$ by successive multiplication 
with appropriate factors $\alpha_{j k}$. The
first two steps are given in 
\begin{align} 
  & \oalphao(P_0 \! + \! Q_0 \! + \! N_0| N_3^0, N_4^0) \
  \oalphao(P_0 \! + \! Q_0 \! + \! N_0 \! + \! N_3^0| N_3^1, N_4^0)  \nonumber \\
  & = \oalphao(P_0 \! + \! Q_0 \! + \! N_0| N_3^0 + N_3^1, N_4^0)
  \, , \nonumber \\
  & \oalphao(P_0 \! + \! Q_0 \! + \! N_0| N_3^0 + N_3^1, N_4^0) \ 
  \oalphao(P_0 \! + \! Q_0 \! + \! N_0 + \! N_3^0 + \! N_3^1| N_3^2, N_4^0)
  \nonumber \\
  & = \oalphao(P_0 \! + \! Q_0 \! + \! N_0| N_3^0 \! + \! N_3^1 +N_3^2, N_4^0) \, .
\end{align}
One then proceeds in the same manner with $N_4$ in the third entries.
By this procedure, one ensures in particular that $M_{0 0} = N_0$. 

\medskip
Let us now have a closer look at the factors $\alpha_{j,k}$. Because of
the support properties of the operators $N_3^j, N_4^k$ for
$j,k = 1,2$, it follows from relation \eqref{e.6.38}
that we can omit $Q_0$  from $\alpha_{j,k}$ for
$j = 1$ as well as $k=1$. Similarly, for $j = 2$ or $k = 2$
we can omit $P_0$. Thus the resulting terms are again products of
phases depending only on $N,P$, respectively $N,Q$. 
There remains the case $j = k = 0$. Recalling that $M_{0 0} = N_0$,
we apply relation \eqref{e.6.34} and get
\begin{align} 
\alpha_{0 0} & = \oalphao(N_0 \! + \! P_0 \! + \! Q_0|N_3^0, N_4^0)  \\
& =\oalphao(N_0 \! + \! Q_0|P_0 \! + \! N_3^0, N_4^0) \, 
  \oalphao(N_0 \! + \! Q_0|P_0, N_4^0)^{-1} \nonumber \\ 
& =\oalphao(N_0 \! + Q_0 \! + N_3^0|P_0, N_4^0) \,
  \oalphao(N_0 \! + \! Q_0|N_3^0, N_4^0) \,
  \oalphao(N_0 \! + \! Q_0|P_0, N_4^0)^{-1}  \! . \nonumber 
\end{align}
Again by relation \eqref{e.6.38},
we can omit $Q_0$ in the first and the third factor
of the latter product. So $\alpha_{0 0}$ also factors into
a product of phases depending only on $N,P$, respectively
$N,Q$. So, to summarize, we succeeded in proving that
there exist phases $\beta_1(P,N)$, $\beta_2(Q,N)$, involving 
decompositions of their arguments
depending only on the given regions $\Reg_1, \dots, \Reg_4$, such that
\be
\beta(P + Q + N) =  \beta_1(P,N) \, \beta_2(Q,N) \, .
\ee
Making use of this equality also for the functional $P = 0$,
respectively $Q = 0$, it is straight forward to
verify relation~\eqref{e.6.39}, completing its proof.  \qed 
\end{proof}

By iteration of this argument, one can trivialize the
extended phase factors $\oalphao(N|P,Q)$ for admissible
triples $P,Q,N \in \Qc$, where $P,Q$
have their supports in any given finite number of 
pairs of disjoint compact regions. In order to
cover all such triples, we make use of Tychonoff's
theorem. Let $\fP$ be any finite collection of pairs
$\Oc^\prime \times \Oc^{\prime \prime}$
of disjoint compact subsets of $\Mc$. We denote by  
$\fB_{\fP}$ the set of maps
$\beta : \Qc \rightarrow \TT$ which trivialize
$\oalphao$ for the given subsets. Recalling the
definition of $\delta \beta$, cf.\ \eqref{e.6.35}, one has
\be
\oalphao(N|P,Q) \, \delta \beta(N|P,Q)^{-1} = 1 
\ee
if $\supp P \times \supp Q \subset \Oc^\prime \times \Oc^{\prime \prime}
\in \fP$. 
We have shown that the sets $\fB_{\fP}$ are not empty and
it is also clear that $\fB_{\fP_1} \subset \fB_{\fP_2}$
if $\fP_1 \supset \fP_2$. Let
\be
\fB \doteq \bigcap_\fP \fB_\fP \, .
\ee
This set is non-empty. Because, otherwise,  
due to the compactness of the set of maps
$\beta : \Qc  \rightarrow \TT$ with respect to the topology of pointwise
convergence (Tychonoff's Theorem), already a finite intersection
had to be empty, which has been excluded. 
Every $\beta \in \fB$ trivializes $\oalphao$. We also
note that different
elements differ by a local functional, \ie 
a map $\gamma : \Qc  \rightarrow \TT$ satisfying
$\delta \gamma = 1$. We have thus established the
following proposition.
\begin{proposition} \label{p.6.7} 
  Let 
  $\oalphao(N|P,Q)$ be the extended phases for admissible triples
  $P,Q,N \in \Qc$ satisfying
  $\supp P \cap \supp Q = \emptyset$. 
  There exists a function $\beta : \Qc \rightarrow \TT$
  such that 
 \be \nonumber
 \oalphao(N|P, Q)= \beta(P + N)^{-1} \, \beta(N) \, 
 \beta(Q + N)^{-1} \, \beta(P + Q + N) \, .
\ee 
\end{proposition}

\medskip
As was shown in Proposition \ref{p.6.5},  
the phases $\alpha(N|P,Q)$ coincide with the restriction of 
$\oalphao(N|P,Q)$ to their domain, \ie on admissible triples
$P,Q,N \in \Qc$ satisfying $P \underset{N}{\succ} Q$.
Thus they can be trivialized, so 
the following corollary obtains. It completes the proof 
of Theorem \ref{t.6.2}. 

\begin{corollary}
  Let $\alpha(N|P,Q)$ be the causal phases,
  introduced in Proposition~\ref{p.5.1} for admissible triples
  $P,Q,N \in \Qzw$ satisfying $P \underset{N}{\succ} Q$.
  There exists a function $\beta: \Qzw \rightarrow \TT$
  such that
  \be \nonumber 
  \alpha(N|P,Q) = \beta(P + N)^{-1} \beta(N) \beta(Q + N)^{-1}
  \beta(P + Q + N) \, .
  \ee     
\end{corollary}  

We conclude this section with a remark on the
covariance properties of our construction.
As noted at the end of the preceding section, the
unitaries $U(\lambda) S(\biP) U(\lambda)^{-1}$ induce
automorphisms of the Weyl operators,
for any $\biP \in \Pc$ and Poincar\'e transformation $\lambda$. 
They exhaust the unitaries for perturbations $\biP_\lambda$ 
satisfying the standing assumption for any given 
time direction and they also satisfy the corresponding causal
factorization condition.
A fully covariant description would require, however, that
the phase factors $\beta$ in the preceding corollary
can be chosen to be Poincar\'e invariant for the given
(Poincar\'e invariant) $\alpha$. It is an open problem whether
such a choice  exists.

\section{Conclusions}

In this article we have extended the framework of dynamical C*-algebras
for quantum field theories on Minkowski space \cite{BF19}, 
admitting also kinetic perturbations. The novel feature appearing in
this extended framework is the influence of kinetic perturbations on the causal
factorization relations of the unitary operators, describing their
impact on states. Whereas these operators still
generate a local, covariant net on Minkowski space, labelled by
their support regions, the causal relations between them 
are affected.  This is due to the fact that they describe
the propagation of fields in distorted spacetimes. As a matter of
fact, this feature imposes restrictions on the admissible perturbations,
put  down in our standing assumption. 
They reflect the idea that the kinetic perturbations are caused
by gravitational effects on the fields. In accordance with
this idea, we have shown that the perturbed fields
satisfy wave equations and commutation relations on
locally perturbed Minkowski spaces. 

\medskip
The unitary operators describing these perturbations
are well defined at the level of abstract C*-algebras, which 
admit an abundance of states and corresponding Hilbert
space representations. Yet it is not clear from the outset that 
there exist also states, describing situations of physical
interest, such as a vacuum and its local excitations,
or equilibrium states. 
As a matter of fact, a comprehensive representation theory of
dynamical C*-algebras is the missing corner stone 
in a rigorous proof that interacting quantum field theories
exist in four spacetime dimensions \cite{BF19}. As was already
mentioned, perturbation theory is of little use in this 
context since it cannot converge in the presence
of kinetic perturbations, due to their impact on the
underlying causal structure and resulting modifications of
commutation relations. Thus a non-perturbative approach to this
problem is needed.

\medskip
As a step into that direction, we have considered 
the subset of perturbations, which are at most quadratic in the underlying
field. These perturbations do not describe self-interactions
of the field, but comprise its interaction with the
spacetime background and perturbations of its mass. 
Previous results by Wald \cite{Wald} had settled the existence of
corresponding unitary operators and resulting local nets of C*-algebras 
on Fock space. But a proof that by adjustment of their
phase factors there exist also operators
which satisfy the causal factorization relations did not
yet exist. In fact, it turned out to be surprisingly involved.

\medskip 
A direct construction of
such unitary operators would have required the
development of a non-perturbative renormalization scheme for
time-ordered exponentials. We have therefore taken here a different,
still cumbersome path. Adopting methods from  
cohomology theory, we have shown that the ambiguous phase factors
of the unitary operators can be fixed in a manner such 
that they satisfy the causal factorization equations,
\ie there are no cohomological obstructions in that respect. 
It completed our proof that the restricted dynamical
algebra is represented on Fock space in any number of spacetime dimensions.
This observation provides further evidence to the effect that our
novel algebraic approach to the construction of quantum field theories
is viable. 

\begin{appendix}
\section*{Appendix}
In this appendix we determine perturbations of the metric $\eta$ in
Minkowski space $\Mc$ which keep it globally hyperbolic, so that the
hypersurfaces $t=\const$ (for a fixed time coordinate) are still
Cauchy surfaces and the time coordinate is positive timelike
with regard to the perturbed metric $g$.
We also analyze in some detail their inverses,
which enter in the corresponding hyperbolic differential operators.
We will thereby justify our standing assumption and exhibit 
increasing families $\Kcc$ of perturbations,
labelled by the velocity of light $c \geq 1$, which enter in our analysis. 

\medskip 
Let $g$ be any such metric. We use the split into time and space and
describe $g$ by a block matrix 
\[ \label{e.a.1}
g = \left(\begin{array}{cc}
            g_{00}    & \bg \\
            \bg^T & - \, \biG
            \end{array}\right) \, , \tag{A.1}
\]
where $\bg$ is a $(d-1)$-vector 
and $\biG$ is a spatial $(d-1) \times (d-1)$-matrix.
According to the conditions on $g$, the chosen time
coordinate is still positive timelike, thus $g_{00} > 0$, and 
spatial vectors are still spacelike,
so $\biG$ has to be positive definite. The lightcone   
$\Vc_+(g)$ fixed by $g$ at any given point in $\Mc$
is determined by the equation for the corresponding
lightlike directions, $v = (1, \biv) \in \RR^d$,  
\be \label{e.a.2}
0 = g(v,v) = g_{00} + 2 \langle\biv,\bg\rangle -
\langle \biv,  \biG \biv \rangle \, .   \tag{A.2}
\ee
Since $\biG \geq \| \biG^{-1} \|^{-1} 1$ one finds that 
\be
|\biv|^2 \| \biG^{-1} \|^{-1} - 2 | \biv |  | \bg |
\, \leq \, g_{00} \, .
\tag{A.3}
\ee
It follows that the velocity of light, determined by $g$, satisfies the bound
\be
| \biv | \leq 
c \doteq
\Big( \sqrt{g_{00} /  \|\biG^{-1} \| + |\bg|^2} +
|\bg| \Big) \|\biG^{-1}\| \, .
\tag{A.4}
\ee
Thus one has the inclusion of light cones 
$\Vc_+(g) \subset \Vc_+(\eta^{c})$,
where the latter lightcone is of Minkowski type, 
\[
\Vc_+(\eta^{\, c}) = \{ (t,\bix) \in \RR^d \, | \,
t>0, \, c^2t^2 - \bix^2 > 0 \} \, , \quad c > 0 \, . \tag{A.5}
\]

\medskip 
Next, we determine the inverse metric. 
Using again the split
into time and space coordinates, we represent $g^{-1}$ also as a block matrix 
\[
g^{-1}=\left(\begin{array}{cc}
            g^{00}      & \bih\\
            \bih^T & - \biH
            \end{array}\right) \tag{A.6}
\]
and obtain by an elementary computation  
\begin{align*}
g^{00} & =(g_{00}+\langle\bg,\biG^{-1}\bg\rangle)^{-1}\ , \\
\bih & =g^{00} \, \biG^{-1}\bg  \, , \\
\biH & = \biG^{-1} - (g^{00})^{-1} \, |\bih\rangle \langle\bih|  \, . \tag{A.7}
\end{align*}
The conditions on $g$ can now also be formulated in terms of
conditions on $g^{-1}$, namely $g^{00}>0$ and
$\biH$ is to be positive definite.

\medskip
The kinetic perturbations $P$, considered in the main
text, are described by differential operators with principal
symbols $p$, which in the chosen coordinates are given by 
\[ \label{e.a.5}
 p =\left(\begin{array}{cc}
            p^{00}      & \bip \\
            \bip^T &    - \biP
            \end{array}\right) \, . \tag{A.8}
\]
Putting $\widetilde{g}_P \doteq (\eta + p)$, the  
corresponding metric $g_P$ on Minkowski space is 
given by the equation
$|\mathrm{det} g_P|^{-1/2}g_P = \widetilde{g}_P^{\, -1}$ (for $d>2$).
So our constraints on the admissible metrics 
imply that $(1 + p^{00}) > 0$ and that the matrix
$(1 + \biP)$  is positive definite. 
These conditions agree with our standing assumption, 
characterizing the principal symbols of admissible
perturbations. 

\medskip
It is apparent that any convex combination of admissible
principal symbols $p$  
is again admissible. We restrict 
the admissible symbols to compact, convex 
subsets in order to control the size of the lightcones
determined by the corresponding metrics~$g_P$ in Minkowski space. 
Given $0 < \varepsilon \leq 1$, we consider perturbations with
principal symbols satisfying 
\be \label{e.a.9}
\varepsilon \leq 1 + p^{00} \leq \varepsilon^{-1} \, , \quad 
\varepsilon \, 1 \leq
 1 + \biP  \leq \varepsilon^{-1} 1  \, . \tag{A.9} 
\ee
We also require that the length $|\bip|$ is bounded by
$\varepsilon^{-1}$. These conditions characterize 
compact convex subsets of principal symbols. Since 
$p = 0$ is contained in any set, they are also contractible. 

\medskip 
In analogy to relation \eqref{e.a.2}, one can determine now the  
momentum space light cones $\Vc_+(\widetilde{g}_P)$  
fixed by the data in relation \eqref{e.a.9}.  
By a similar computation as above
one finds that the vectors $(1, \bik)$ are contained in these cones if
$| \bik | \leq (\sqrt{2} - 1) \varepsilon^2$. Thus the cones contain the 
momentum space lightcones for the Minkowskian metric $\eta^{c(\varepsilon)}$
with velocity of light
\be
c(\varepsilon)= (\sqrt{2} + 1) \,  \varepsilon^{-2} \, . \tag{A.10}
\ee
For the dual lightcones in position space
$\{ x \in \RR^d : xp > 0, \ p \in \Vc_+(\widetilde{g}_P) \}$
we get the opposite inclusion.
Hence the metrics $g_P$ associated with the
(for the given $\varepsilon$) restricted principal symbols
$p$ are dominated on all of Minkowski space 
by the Minkowskian metric $\eta^{\, c(\varepsilon)}$.
In particular, these metrics comply with our initial
constraints on amissible metrics.
Because of the relevance of the value of the
parameter $c$ in the main text, we
denote the corresponding compact, convex and contractible 
sets of principal symbols by $\Kcc$. They increase with
increasing $c$ and exhaust the set of all admissible
principal symbols. 

\end{appendix}

\begin{acknowledgements}
We are grateful to Valter Moretti and Rainer Verch for information
on the status of the problem of Haag duality in curved spacetimes.
DB~also thanks Dorothea Bahns and the Mathematics Institute
of the University of G\"ottingen for their generous hospitality.  
\end{acknowledgements}

\end{document}